\documentclass[draftcls,11pt,onecolumn,letterpaper]{IEEEtran}
\IEEEoverridecommandlockouts
\usepackage{enumerate}
\usepackage{float}
\usepackage{color}
\usepackage{graphicx}
\usepackage{amsmath} % assumes amsmath package installed
\usepackage{amssymb,algorithm}  % assumes amsmath package installed
\newtheorem{Theorem}{Theorem}[section]
\newcommand{\trace}{\mathop{\mathrm{trace}}}
\newcommand{\beq}{\begin{equation}}
\newcommand{\eeq}{\end{equation}}

\newcommand{\bq}[1]{\begin{equation} \label{#1}}
\newcommand{\eq}{\end{equation}}
\newcommand{\bed}{\begin{displaymath}}
\newcommand{\eed}{\end{displaymath}}
\newcommand{\bea}{\bed\begin{array}{rl}}
\newcommand{\eea}{\end{array}\eed}
\newcommand{\ad}{&\!\!\!\disp}

\newcommand{\bE}{{\mathbf{E}}}
\newcommand{\tr}{{\hbox{tr}}}
\def\disp{\displaystyle}
\newcommand{\e}{\varepsilon}

\newcommand{\barray}{\begin{array}{ll}}
\newcommand{\earray}{\end{array}}
\newcommand{\M}{{\cal M}}
\newcommand{\cd}{(\cdot)}
\newcommand{\rr}{{\mathbb R}}
\renewcommand{\hat}{\widehat}
\renewcommand{\tilde}{\widetilde}
\renewcommand{\bar}{\overline}
\newcommand{\qed} {{$\hfill\blacksquare$}}
\newcommand{\graph}{G}      % Graph
\newcommand{\pdup}{r}       % probability of Duplication step
\newcommand{\pdel}{q}       % probability of Deletion step
\newcommand{\p}{p}          % probability of edge-dupilication step
\newcommand{\E}{\mathbf{E}}
\newcommand{\esa}{\varepsilon} % time scale for Stochastic Approximation
\newcommand{\emc}{\rho}   % time scale for Markov chain
\newcommand{\bg}{\bar{g}}   % exoected PMF
\newcommand{\obs}{y}        % Observations
\newcommand{\noise}{e}      % noise
\newcommand{\g}{g}          % PMF
\newcommand{\hg}{\hat{\g}}  % estimated PMF by SA
\newcommand{\degree}{f}     % degree sequence (nomber of nodes with specific degree)
\newcommand{\transition}{L} % Transition matrix of true parameter
\newcommand{\A}{A}          % Transition matrix of Markov chain
\newcommand{\tim}{n}        % Discrete Time
\newcommand{\ct}{t}         % Continous time
\newcommand{\baf}{\bar{f}}  % expected f
\newcommand{\mc}{\theta}    % Markov chain
\newcommand{\ttrue}{B}      % Transtion matrix of true parameter
\newcommand{\diag}{\mathop{\mathrm{diag}}}
\newcommand{\tg}{\tilde{\g}}% the difference between sample path and the expected one
\newcommand{\Et}{\E_\tim}   % the expectation with respect to time
\newcommand{\ser}{\nu}      % scaled error
\newcommand{\cov}{\Sigma}   % Covariance of the distance between sample path and expected pmf
\newcommand{\s}{{N_0}}          % size of the pmf
\newcommand{\delay}{\lambda}
\newcommand{\de}{\bar{d}_1}
\newcommand{\dd}{\bar{d}_2}
\newcommand{\noisee}{\omega}

\newtheorem{Corollary}{Corollary}[section]
\newtheorem{Lemma}{Lemma}[section]

\newtheorem{Definition}{Definition}[section]

\title{\LARGE \bf Tracking the Empirical
Distribution of a Markov-modulated Duplication-Deletion Random Graph}

\author{Maziyar Hamdi, Vikram Krishnamurthy,\thanks{Maziyar Hamdi and Vikram Krishnamurthy are with the Department of Electrical and Computer
Engineering, University of British Columbia, Vancouver, Canada. Email:
{\small \{maziyarh,vikramk\}@ece.ubc.ca.}} and George Yin \thanks{George
Yin is with the Department of Mathematics,
Wayne State University, Detroit, US. Email: {\small gyin@math.wayne.edu.}}\thanks{Parts of this work was presented at the 2012 IEEE International Conference on Acoustics, Speech and Signal Processing (ICASSP2012), Kyoto, Japan.}
}

\begin{document}

\maketitle
\thispagestyle{empty}
\pagestyle{empty}

%%%%%%%%%%%%%%%%%%%%%%%%%%%%%%%%%%%%%%%%%%%%%%%%%%%%%%%%%%%%%%%%%%%%%%%%%
%%%%%%%
\begin{abstract}

\hspace{1mm}
This paper considers a Markov-modulated duplication-deletion random graph where at each time instant, one node can either join or leave the network; the probabilities of joining or leaving evolve according to the realization of a finite state Markov chain. The paper comprises of 2 results.  First, motivated by social network applications, we analyze the asymptotic behavior of the degree distribution of the Markov-modulated random graph. Using the asymptotic degree distribution,   an expression is obtained for the delay in searching such graphs. Second, a stochastic approximation algorithm is presented to track empirical degree distribution as it evolves over time. The tracking performance of the algorithm is analyzed in terms of mean square error and a functional central limit theorem is presented for the asymptotic tracking error.

\end{abstract}

%%%%%%%%%%%%%%%%%%%%%%%%%%%%%%%%%%%%%%%%%%%%%%%%%%%%%%%%%%%%%%%%%%%%%%%%%
%%%%%%%
\section{INTRODUCTION}\label{sec:intro}

Dynamic random graphs have been widely used to model social networks, biological networks \cite{duplication} and Internet graphs \cite{complex}. Motivated by analyzing social networks, this paper considers  Markov-modulated dynamic random graphs of the duplication-deletion type which we now describe:%The nodes represent  agents (members) and edges depicts the relation between incident nodes.

Let $\tim=0,1,2,\ldots$ denote discrete time. Let $\mc$ denote a discrete time Markov chain with state space $\{1,2,\ldots,M\}$, evolving according to the $M \times M$ transition probability matrix $\A^{\emc}$ and initial probability distribution $\pi_0$. A Markov-modulated duplication-deletion random graph is parameterized by the 7-tuple $(M,\A^{\emc},\pi_0,\pdup,\p,\pdel,\graph_0)$. Here  $\p $  and $\pdel$ are $M$-dimensional vectors with elements $\p(i)$ and $\pdel(i)  \in [0,1]$, $i=1,\ldots,M$. $\p(i)$ denote the connection probabilities and $\pdel(i)$ denote the deletion probabilities. Also, $\pdup \in [0,1]$ denotes the probability of duplication step and $\graph_0$ denotes the  initial graph at time~$0$. $\graph_0$ can be any finite simple connected graph. For simplicity we assume that $\graph_0$  is a simple connected graph with size $N_0$. The duplication-deletion random graph is constructed as follows:

\begin{algorithm}{}
 At time $\tim$, given the graph $\graph_\tim$ and Markov chain state
$\mc_\tim$,  simulate the following events: \\
\textbf{Step 1: Duplication step}: With probability
$\pdup$ implement the following steps:
\begin{itemize}
\item Choose node $u$ from graph $\graph_\tim$ randomly with uniform
distribution.
\item  \textit{Vertex-duplication}: Generate a new node $v$.
\item \textit{Edge-duplication}:
\begin{itemize}
\item Connect node $u$  to node $v$. (A new edge between $u$ and $v$ is
added to the graph.)
\item Connect each neighbor of node $u$  with probability
$\p(\mc_\tim)$  to  node $v$. These connection events are statistically
independent.
\end{itemize}
\end{itemize}

\textbf{Step 2: Deletion Step}: With probability
$\pdel(\mc_\tim)$ implement the following step: \begin{itemize}
 \item \textit{Edge-deletion}: Choose node $w$ randomly from
$\graph_\tim$ with uniform distribution. Delete node $w$ and all edges
 connected to node $w$ in graph $\graph_\tim$.
 \item \textit{Duplication Step}:  Implement Step 1.

 \end{itemize}
 \textbf{Step 3}: Denote the resulting graph as
$\graph_{\tim+1}$.\\Generate Markov state $\mc_{\tim+1}$ using transition
matrix $\A^{\emc}$.\\
 \textbf{Step 4: Network Manager's Diagnostics}: The network manager
computes the estimates of the expected degree distribution.
 Denote the resulting graph as $\graph_{\tim+1}$. \\
 Set $\tim \rightarrow \tim+1$ and go to Step 1. \caption{Markov-modulated Duplication-deletion Graph parameterized by
$(M,\A^{\emc},\pi_0,\pdup,\p,\pdel,\graph_0)$}
 \label{alg0}
\end{algorithm}
For convenience in our analysis, assume that a  node generated in the duplication step cannot be eliminated in the deletion step immediately after its generation. Also to prevent the isolated nodes, assume that the neighbor of a node with degree one cannot be eliminated in the deletion step. The duplication step (Step 2) is purely for convenience - it ensures that the graph size does not decrease. The Markov-modulated random graph generated by Algorithm~\ref{alg0} mimics social networks where the interaction between nodes  evolves over time due to underlying dynamics such as seasonal variations (e.g., the high school friendship social network  evolving over time with different winter/summer dynamics). In such cases, the connection/deletion probabilities $\p,\pdel$ evolve with time. Algorithm~\ref{alg0} models these time variations as a finite state Markov chain $\mc_\tim$ with transition matrix $\A^{\emc}$.
\subsection*{Context: Why is the degree distribution important?}\label{subsec:con}

The expected degree distribution yields useful information about the connectivity of the random graph. For example, if a majority of nodes in the random graph have relatively high degrees, the graph is highly connected and a message can be transferred between two arbitrary nodes with shorter paths. However, if a majority of nodes have smaller degrees then for transmitting a message throughout the network, longer paths are needed, see \cite{social}. Also, the degree distribution can be used to determine the existence of ``giant component"\footnote{A giant component is a connected component with size $O(n)$ where $n$ is the total number of vertices in the graph.}. The existence of a giant component has important implications in social networks in terms of modeling information propagation in a social network and in human disease modeling \cite{Newman1,Newman2,Nature}. If the average degree of a random graph is strictly greater than one then with probability one there exists a unique giant component \cite{complex} and the size of this component can be computed from the expected degree sequence. The average degree and the size of giant component is computed at each time as a measure of connectivity by the monitoring node. Another application of tracking the expected degree distribution is to estimate adaptively the ``\textit{searchability}'' of the network. The searchability of a  social network \cite{CSN} is  the average number of nodes that need to be accessed  to reach another node. In this paper, we track the searchability of the network by means of tracking the expected degree distribution at each time.

\subsection*{Main Results and Paper Organization:}
{\em Notation:} At each time $\tim$, let $N_\tim$ denote the number of nodes of graph $\graph_\tim$. Also, let $f_\tim(i)$ denote the number of vertices of graph $\graph_\tim$ with degree $i$. Clearly $\sum_{i\geq 1} f_\tim(i) = N_\tim$. Define the ``empirical vertex degree distribution'' as
\beq \label{eq:sample}
g_\tim(i) = \frac{f_\tim(i)}{N_\tim}, \quad \text{for } 1\leq i\leq N_\tim.\eeq
 Note that $g_\tim(i)$  can be viewed as a probability mass function since $g_\tim(i) \geq 0$ and $\sum_i g_\tim(i) = 1$.
 Let  $\bg_\tim = \E\{g_\tim\}$ denoted the ``expected vertex degree distribution'' where $\g_\tim$ is the empirical degree distribution defined in (\ref{eq:sample}).

 Given the above Markov-modulated random graph, this paper presents three main results.
 \\
 \textbf{Result 1:}\textit{ Asymptotic Degree Distribution Analysis of fixed size Markov-modulated duplication-deletion random graph}

%Sec.\ref{sec:power}  analyzes the expected degree distribution for two different scenarios:

%\textit{1-a. Infinite duplication-deletion random graph without Markovian dynamics:}
%
% Sec.\ref{subsec:pl}  investigates the dynamics of the graph generated according to Algorithm \ref{alg0} with $\pdup = 1$ and when there are no Markovian dynamics, that is, $M = 1$. Since $\pdup = 1$ for $\tim \geq 0$, $\graph_{\tim+1}$ has one more vertex compared to $\graph_\tim$. In particular, since $\graph_0$ is an empty set, $\graph_\tim$ has $\tim$ nodes, that is, $N_\tim = \tim$. Theorem~\ref{theo:pl} proves that the expected node degree distribution $\bg_\tim$ satisfies a power law as $n \rightarrow \infty$. That is, $$\log \bg_\tim(i)  = \alpha - \beta \log i \quad  \text{ as } \tim\rightarrow \infty$$ where $\alpha$ and $\beta$ are non-negative real numbers.  The power law component, $\beta$, satisfies
%    \begin{align}
%        (1+\pdel)( \p^{\beta - 1 }+\p\beta - \p ) =1 + \beta \pdel.
%    \end{align}
%where $\p$ and $\pdel$ are the probabilities defined in Algorithm~\ref{alg0}.
% The above result slightly extends~\cite{duplication,bebek25} where only a duplication model is considered. Theorem~\ref{theo:pl} parametrizes the degree distribution of the infinite duplication-deletion random graph without Markovian dynamics generated by Algorithm~\ref{alg0} by the power law component.

 Consider the sequence of  finite duplication-deletion random graphs $\{\graph_\tim\}$, generated by Algorithm \ref{alg0} with $\pdup = 0$. Clearly the number of vertices in the graph generated by Algorithm~\ref{alg0} with $\pdup = 0$  satisfies $N_\tim  = N_0$ for $ n = 1,2,\ldots$ (The size of random graph is fixed.).  Assume that the Markov chain  $\mc_\tim$ evolves according to a slow transition matrix $\A^{\emc} = I + \emc Q$, where $Q$ is a generator matrix and $\emc$ is a small positive constant. A novel degree distribution analysis is provided for the fixed size Markov-modulated duplication-deletion random graph in Sec.\ref{sec:markov}. Theorem~\ref{theo:mu} shows  that for each $\mc_\tim = \mc$, the expected degree distribution of the finite random, $\bg(\mc)$, can be computed from (\ref{eq:gbar}). %Numerical results show that when the size of the graph ($N_0$) is sufficiently large, the fixed size random graph with $M=1$ also satisfies a power law.
 % Maziyar: This does not make sense - if M=1, then it is no longer Markov modulated. Also I dont understand the statement above that Theorem 2.2 proves ... What is the
% main implication of the result? One possibility is to simply remove the entire 1-b and Sec.IIB.

The asymptotic degree distribution analysis allows us to investigate  the searchability and connectivity of the random graph generated according to Algorithm~\ref{alg0} as described in Sec.\ref{sec:markov}. Also, using the asymptotic degree distribution, the existence and size of the giant component in the random graph  can be explored.\\\\
\textbf{Result 2:}\textit{ Tracking the Empirical Degree Distribution}

In Sec.\ref{sec:pmf}, we address the following two questions: \begin{itemize}
\item \textit{How can a network manager estimate (track) the empirical degree distribution using a stochastic approximation algorithm without knowledge of Markovian dynamics?}
\item \textit{How good is the estimate $\hg_\tim$ generated by the stochastic approximation algorithm (\ref{eq:sa}) when the random graph evolves according to Algorithm~\ref{alg0}?}
\end{itemize}
 In Sec.\ref{sec:pmf}, we propose a  stochastic approximation algorithm to estimate the degree of each node in random graph which can be modeled by Algorithm~\ref{alg0}. Consider the finite Markov-modulated duplication-deletion random graph generated by Algorithm \ref{alg0} with 7-tuple $(M, A^\emc,\pi_0,\p,\pdel,\pdup,\graph_0)$ where $\pdup = 0$. Suppose at each time $\tim$, noisy measurements, $\obs_\tim$ the empirical distribution of $\g_\tim$ are obtained by the administrator of the social network. The network manager does not have information about the Markovian dynamics and deploys a non-parametric stochastic approximation algorithm to estimate the expected vertex degree distribution. More precisely, given these measurements $\obs_\tim$, $\tim=1,2,\ldots$, the network administrator aims to estimate the time varying expected vertex distribution $\bg(\mc_\tim)$. It deploys the following constant step size stochastic approximation algorithm:
\beq\label{eq:sa}
  \hg_{\tim+1} = \hg_\tim +  \esa \left [ \obs_{\tim+1} - \hg _\tim
\right ] \eeq
Here $\esa > 0$ denotes a small positive step size. Eq. (\ref{eq:sa}) is merely an exponentially discounted empirical distribution of the noisy node degree.  Let $\tg_\tim = \hg_\tim - \E\{\bg(\mc_\tim)\}$ denote the tracking error of the estimate of the empirical distribution of node degree. We present three results regarding the tracking performance of the degree distribution of the random graph: \begin{itemize}
 \item {\em 2-a. Mean square error analysis:} Theorem~\ref{theo2} in Sec.\ref{subsec:bound} shows that the mean squared of tracking error (the distance between $\E\{\bg(\mc_\tim)\}$ and the estimated probability mass function (PMF)  $\hg_\tim$) is of order of $O\left(\esa + \emc + \frac{\emc^2}{\esa}\right)$. (Recall $\esa$ is the step size of the stochastic approximation algorithm and $\emc$ parameterizes the speed of the underlying un-observed Markovian dynamics). Derivation of this result uses error bounds on  two-time scale Markov chains and perturbed Liapunov function methods.
 \item {\em 2-b. Weak convergence analysis:} Theorem~\ref{theo3} in Sec.\ref{subsec:ode} shows that the asymptotic behavior of the stochastic approximation algorithm (\ref{eq:sa}) converges weakly  to the solution of a switched Markovian ordinary differential equation
\begin{equation}
 \frac{d\hg(\ct)}{d\ct} = -\hg(\ct) + \bg(\mc(\ct)), \quad\hg(0) = \hg_0.
\end{equation}
\item {\em 2-c. Functional central limit theorem for scaled tracking error:} How can the tracking error in the empirical distribution estimate be quantified? Sec.\ref{subsec:er} investigates the asymptotic behavior of the scaled tracking error. Similar to \cite{KTY09}, it is shown that the interpolated scaled tracking error (between the expected and the estimated PMF) converges weakly to the solution of a switching diffusion. Let $\ser_k = \frac{\hg_k - \mathbf{E}\{\bg(\mc_k)\}}{\sqrt{\esa}}$ denote the scaled tracking error.
 Theorem~\ref{theo4} in Sec.\ref{subsec:er} proves that under reasonable conditions, the interpolated sequence of iterates, $\ser^\esa(\ct) = \ser_k$ for $k \in [k\esa, (k+1)\esa)$ converges weakly to the solution of the following Markovian switched diffusion process
\begin{equation}\label{eq:sig}
d\ser(\ct) = -\ser(\ct) d\ct + \left(\cov^{\frac{1}{2}}(\mc(\ct))\right)d\omega,
\end{equation}
where $\omega(\cdot)$ is  an $\rr^{N_0}$-dimensional standard Brownian motion
%a real-valued, standard Brownian motion
and  $\cov(\mc) \in \rr^{N_0\times N_0}$ is the covariance matrix. Eq. (\ref{eq:sig}) (and Theorem~\ref{theo4}) are functional central limit theorems. The dynamics of the error in (\ref{eq:sig}) follow a Markov-modulated diffusion process. The covariance $\cov(\mc(\ct))$ for large $\ct$ is used as a measure for the asymptotic convergence rate of the tracking algorithm.
\end{itemize}

Note that the Markovian assumption only appear in our analysis, the stochastic approzimation algorithm (\ref{eq:sa}) does not assume knowledge of the underlying Markov chain. (In \cite{BMP, protter} this analysis falls under the class of analysis of a stochastic approximation algorithm with a Markovian hyperparameter.)

\textbf{Result 3:\textit{ Power law component for infinite duplication-deletion random graph without Markovian dynamics}}

Sec.\ref{subsec:pl} extends the results of Sec.\ref{sec:markov} and investigates the dynamics of the graph generated according to Algorithm \ref{alg0} with $\pdup = 1$ and when there are no Markovian dynamics, that is, $M = 1$. Since $\pdup = 1$ for $\tim \geq 0$, $\graph_{\tim+1}$ has one more vertex compared to $\graph_\tim$. In particular, since $\graph_0$ is an empty set, $\graph_\tim$ has $\tim$ nodes, that is, $N_\tim = \tim$. Theorem~\ref{theo:pl} proves that the expected node degree distribution $\bg_\tim$ satisfies a power law as $n \rightarrow \infty$. That is, $$\log \bg_\tim(i)  = \alpha - \beta \log i \quad  \text{ as } \tim\rightarrow \infty$$ where $\alpha$ and $\beta$ are non-negative real numbers.  The power law component, $\beta$, satisfies
    \begin{align}
        (1+\pdel)( \p^{\beta - 1 }+\p\beta - \p ) =1 + \beta \pdel.
    \end{align}
where $\p$ and $\pdel$ are the probabilities defined in Algorithm~\ref{alg0}.
 The above result slightly extends~\cite{duplication,bebek25} where only a duplication model is considered. Theorem~\ref{theo:pl} parametrizes the degree distribution of the infinite duplication-deletion random graph without Markovian dynamics generated by Algorithm~\ref{alg0} by the power law component. Theorem~\ref{theo:pl} allows us to explore the searchability of the network and also the existence and size of the giant component of the infinite duplication-deletion random graph without Markovian dynamics.
 \subsection*{Related Works:} We refer to \cite{KY03,YKI04} for a
comprehensive development of stochastic approximation algorithms. Here,
the related literature on dynamic social networks is reviewed briefly.
The evolution of random graphs is investigated in several
papers,\cite{erdos2,evolution-nature}. The book \cite{durrett} provides a
detailed expositions of random graphs. The model of Pastor-Satorras et
al.\cite{bebek25} makes the basis for the model which is studied and
generalized in this paper. In this model, at each time step, a new node
joins the network. In the literature, it has been shown that the degree
distribution of such network satisfies \textit{power
law}\cite{bebek19,bebek30}. In random graphs which satisfy the power law,
the number of nodes with an specific degree depends on a parameter called
power law component. A general complex graph generated by any arbitrary
pure duplication, may not satisfy the power law. The power law
distribution is satisfied in many other networks such as WWW-graphs,
peer-to-peer networks, phone call graphs and various massive social
networks (e.g. Yahoo, MSN,
Facebook)\cite{bebek1,bebek5,bebek11,bebek13,bebek17,bebek21,chung116}.
The power law component describes asymptotic behavior of an online social
network e.g. maximum degree, existence of giant component, diameter of
the graph, and etc. \cite{bebek} provides condition on the evolution of
the graph to satisfy power law and shows that  as a result of having an
edge between nodes $u$ and $v$, the resulting graph satisfies power
law.

\section{Asymptotic Degree Distribution Analysis of the Fixed Size Markov-modulated Random Graph}\label{sec:markov}

This section presents degree distribution analysis of the \textit{fixed size Markov-modulated duplication-deletion random graph}.  %It considers a Markov-modulated graph generated by Algorithm~\ref{alg0} with $\pdup = 0$.
%The number of vertices in the graph generated by Algorithm
%\ref{alg0} with $\pdup = 0$ is always $N_0$. %random
%graph generated by Algorithm~\ref{alg0}. As mentioned in Sec.\ref{sec:intro},
%we consider the following two scenarios:
% \begin{enumerate}
%\item {\em Infinite duplication-deletion random graph without Markovian dynamics:} In this scenario,
%we investigate the random graph generated according to Algorithm
%\ref{alg0} with $\pdup = 1$ and when there are no Markovian dynamics,
%that is, $M = 1$. Since $\pdup = 1$ for  $\tim \geq 0$,
%$\graph_{\tim+1}$ has one more vertex compared to  $\graph_\tim$. In
%particular, since $\graph_0$ is an empty set, $\graph_\tim$ has
%$\tim$ nodes, that is, $N_\tim = \tim$. The main result is Theorem~\ref{theo:pl}.
%    \item {\em Fixed size Markov-modulated duplication-deletion random graph}: Here, we consider a Markov-modulated graph generated by Algorithm~\ref{alg0} with $\pdup = 0$.
%The number of vertices in the graph generated by Algorithm
%\ref{alg0} with $\pdup = 0$ is always $N_0$. The main result is Theorem~\ref{theo:mu}.
%\end{enumerate}
%First, the infinite random graph is considered and we state our first
%main result namely Theorem~\ref{theo:pl}. Then, we investigate the Markov
%modulated random graph and show that the evolution of the expected degree
%distribution in this graph can be modeled by another Markov chain.
Consider the fixed size Markov-modulated duplication-deletion random graph generated according to
Algorithm~\ref{alg0} with 7-tuple $(M, A^\emc,\pi_0,\p,\pdel,\pdup,
\graph_0)$ where $\pdup =0$. The number of vertices in the graph generated by Algorithm
\ref{alg0} with $\pdup = 0$ is always $N_0$ and the size of the graphs is fixed. Recall from Sec.\ref{sec:intro}, the state space of $\{\mc_\tim\}$ is denoted as
\begin{equation}
\mathcal{M} = \{1,2,...,M\},
\end{equation}
and the transition probability matrix of
$\mc_\tim$ is
\begin{equation}
\label{eq:A}
A^{\emc} = I + \emc Q.
\end{equation}
Here $\emc$ is a small positive real number and so $\mc_\tim$ is a
``slow" Markov chain. $I$ is an $M\times M$ identity matrix, and $Q$ is
an irreducible generator of a continues-time Markov chain. Let
$q_{ij}$ denote the elements of the generator matrix $Q$ such that
\begin{itemize}\item{(A)} $q_{ij} \geq 0 \hspace{2mm}{\rm if} \hspace{2mm}
i\neq j$ and  $\forall i,\hspace{2mm} \sum_{j=1}^{M}q_{ij} = 0.$ For simplicity, we assume that the initial distribution  $\pi_0$ is independent of $\emc$. $Q$ is irreducible\footnote{The assumption of irreducibility implies that there exists a unique
stationary distribution for this Markov chain, $\pi \in \mathbb{R}^{M
\times 1}$ such that \begin{equation}\label{pi}
\pi' = \pi' \A^{\emc}.
\end{equation}}.
\end{itemize}

%Sec.~\ref{subsec:pl} proves that the infinite duplication-deletion random graph without Markovian dynamics satisfies a power law. Since the number of nodes is constant in the fixed size Markov-modulated duplication-deletion random graph, we could not apply the same degree distribution analysis as Sec.~\ref{subsec:pl}. However, the approach used in the proof of Theorem~\ref{theo:pl} can be employed to find the expected degree distribution of the fixed size Markov-modulated duplication-deletion random graph.
Theorem~\ref{theo:mu} below proves that the expected degree distribution of the fixed size markov-modulated duplication-deletion random graph satisfies a recursive equation from which the expected degree distribution can be found.
\begin{Theorem}\label{theo:mu}
Consider the  fixed size Markov-modulated duplication-deletion random graph generated according to Algorithm~\ref{alg0} with 7-tuple $(M, A^\emc,\pi_0,\p,\pdel,\pdup,
\graph_0)$  where $A^\emc = I
+ \emc Q$ and $\pdup =0$. Let $\bg^{\mc}_\tim =
\E\{\g_\tim|\mc_\tim = \mc\}$. The expected degree distribution of nodes
in the fixed size Markov-modulated duplication-deletion random graph, $\bg^{\mc}_\tim$, satisfies the
following recursion
\beq
\label{eq:true1}
\bg^{\mc}_{\tim+1} = (I + \frac{1}{N_0}\transition'(\mc))\bg^{\mc}_{\tim},
\eeq
where  $^\prime$ denotes transpose of a matrix and $\transition(\mc_{\tim})$, with elements defined in (\ref{eq:L}),
is a generator matrix (that is, each row adds to zero and each diagonal element of $\transition(\mc_\tim)$  is negative):
\beq \label{eq:L}
      l_{ji} = \left\{\begin{array}{ll}
       0 & j < i-1 \\
      \pdel(\mc_\tim) \p(\mc_\tim)^{i-1} + \pdel(\mc_\tim) \big(1 + \p(\mc_\tim)(i -
1)\big) &j = i - 1 \\
       i \pdel(\mc_\tim)\p(\mc_\tim)^{i-1} (1 - \p(\mc_\tim)) -\pdel(\mc_{\tim})\big(i + 2 +
\p(\mc_\tim)i \big)  &j = i  \\
      \pdel(\mc_\tim) {{i+1}\choose{i-1}}\p(\mc_{\tim})^{i-1}(1-\p(\mc_{\tim}))^{2} +
\pdel(\mc_\tim)(i+1) & j = i + 1 \\
     \pdel(\mc_\tim)  {{j}\choose{i-1}}\p(\mc_{\tim})^{i-1}(1-\p(\mc_{\tim}))^{j-i+1} &
j > i + 1
          \end{array}\right.\quad \text{for $1\leq i,j\leq N_0$}
       \eeq
\qed

\end{Theorem}\vspace{3mm}
The proof is presented in Appendix \ref{ap:mu}.

Theorem~\ref{theo:mu} shows that the evolution of the expected degree
distribution in a fixed size Markov-modulated duplication-deletion random graph satisfies
(\ref{eq:true1}). Eq. (\ref{eq:true1}) can be re-written as \beq \bg_{\tim+1}^\mc
= \ttrue^{\prime}(\mc)\bg_{\tim}^\mc, \eeq where  $\label{eq:B}
\ttrue(\mc_\tim) = I + \frac{1}{N_0} \transition(\mc_{\tim}).$
Since $\transition(\mc_\tim)$ is a generator,
for sufficiently large $N_0$,
 $\ttrue(\mc_\tim)$ can
be considered as the transition matrix of a Markov chain.
% with
%time scale $\frac{1}{N_0}$. As a result,
Hence, for each state of the Markov chain $\mc_\tim = \mc \in \{1,2,\ldots,M\}$,
there exists a unique stationary distribution $\bg(\mc)$ such that
\beq \label{eq:gbar} \bg(\mc) = \ttrue^{\prime}(\mc) \bg(\mc).\eeq
Therefore from (\ref{eq:gbar}), the expected degree distribution of the fixed size Markov-modulated duplication-deletion random graph can be computed for each state of the underlying Markov chain $\mc_\tim = \mc$. Note that the underlying markov chain $\mc_\tim$ depends on the small parameter $\emc$. The main idea is that although $\mc_\tim$ is time-varying but it is piecewise constant and since $\emc$ is small parameter, it changes slowly over time. Also from (\ref{eq:true1}), the evolution of $\bg_\tim^\mc$ depends on $\frac{1}{N_0}$. Our assumption throughout this paper is that $\emc \ll \frac{1}{N_0}$. This means that the evolution of $\bg_\tim^\mc$ is faster than the evolution
of $\mc_\tim$ or equivalently it can be said that $\bg_\tim^\mc$ reaches its stationary
distribution ($\bg(\mc)$) before the state of $\mc_\tim$ changes.

\subsection*{Example: Searchability of a Network} So far in this section, an asymptotic analysis of the degree distribution was presented for a random graph generated according to Algorithm~\ref{alg0}. We now comment briefly on how the degree distribution can be used to investigate the searchability of the network.
This also motivates the stochastic approximation algorithm presented in Sec.\ref{sec:pmf} as will be described below. The search problem arises in a network when a specific node faces a
problem (request) whose solution is at other node (e.g., delivering a letter to a specific person or  finding a web page with specific information).
Assume \cite{CSN} that on receiving a search request, each node follows the following protocol:  (a) It address the request if it or its neighbors have the solution; otherwise (b) it relays the request to one
of its neighbors chosen uniformly.
The objective is to find the expected search delay, that is, the expected number of steps until the request is addressed.
% when the size of the graph (network) is large.

 %A trivial (but inefficient) approach  is to  navigate the network randomly until reaching the target. The search can be made more efficient by exploiting the network structure as follows \cite{CSN}: any node when receives a request, knows whether it can address it or one of its neighbors can do that. Otherwise it relays it to one of its neighbors that is chosen randomly.

  \begin{Lemma}\label{lem:search} Consider the sequence of  fixed size Markov-modulated duplication-deletion random graph
 obtained by Algorithm \ref{alg0}, $\{\graph_\tim\}$, with $(M, A^\emc,\pi_0,\p,\pdel,\pdup, \graph_0)$ where $A^\emc
= I + \emc Q$ and $\pdup = 0$ and expected
degree distribution $\bg_\tim$. The expected search delay
 is \beq \label{eq:delay} \delay(N_0) =
O\left(\frac{N_0\de}{\dd - \de}\right), \eeq as $\tim \rightarrow
\infty$ where $\de = \sum_{i=1}^{N_0}i\bg_\tim(i)$ and $\dd =
\sum_{i=1}^{N_0}i^2\bg_\tim(i) $.
\end{Lemma}

\begin{proof}
See Chapter 5 of \cite{CSN} and recall that size of the considered random
graph is $N_0$.
\end{proof}
Lemma~\ref{lem:search} implies that, if the empirical degree distribution of the possibly time-varying network can  tracked accurately, then such an estimate can be used to track the searchability of the network. Also, using the estimated degree
distribution and Lemma~\ref{lem:search}, we can address the following
design problem as: \textit{How can  $\p$ and
$\pdel$ in Algorithm~{\ref{alg0}} be chosen so that the average delay does not
exceed a threshold?}\\
Using the stochastic approximation algorithm in (\ref{eq:sa}) (see Sec.\ref{sec:pmf} below for the convergence proof), we can
estimate the expected degree distribution, $\hg_\tim$, and from that, we can
compute  $\de$ and $\dd$. Then, from Lemma~\ref{lem:search} we can find
the measure of searchability and compare it with the maximum acceptable
average delay and modify the parameters of Algorithm~\ref{alg0}
accordingly. We illustrate searchability in numerical examples given in Sec.\ref{sec:num}.
\section{Estimating (Tracking) the Degree Distribution of the Fixed Size Markov-modulated Duplication-deletion Random Graph }\label{sec:pmf}
In Sec.\ref{sec:markov}, a degree distribution analysis is provided for the fixed size Markov-modulated duplication-deletion random graph generated by
Algorithm \ref{alg0} with 7-tuple
$(M,A^{\emc},\pi_0,\pdup,\p,\pdel,\graph_0)$, where
$\pdup = 0$, $\graph_0$ is a simple connected grapeh of size $N_0$  and
$A^{\emc}$ is defined in (\ref{eq:A}). In this section we assume that the
empirical degree distribution of this graph, $\g_\tim$, is observed in
noise by a network administrator. How can the network administrator track
the expected degree distribution of the fixed size Markov-modulated duplication deletion random graph
without knowing the dynamics of the graph?
%In Sec.\ref{sec:markov}, we considered the Markov-modulated random graph
%generated by Algorithm \ref{alg0} with 7-tuple
%$(M,A^{\emc},\pi_0,\pdup,\p,\pdel,\graph_0)$, where $\pdup$ is
%defined in (\ref{eq:r}), $\graph_0$ is empty set and $A^{\emc}$ is
%defined in (\ref{eq:A}).
Suppose that the vertex distribution $\degree_\tim$ generated according
to Algorithm \ref{alg0} is measured in noise by the administrator of the
social network. That is, the measurement is
\beq \hat{\degree}_\tim = \degree_\tim + \noisee_\tim.   \eeq
Here, at each time $\tim$, the elements $\noisee_\tim(i)$  of the noise
vector are integer-valued  zero mean random variables and
$\sum_{i\geq1} \noisee_\tim(i) = 0$. The zero sum assumption ensures that
$\hat{\degree}_\tim$ is a valid empirical distribution.
In terms of the empirical vertex distribution, we can rewrite this
measurement process as
$$\obs_\tim(i) = \frac{\hat{\degree}_\tim(i)}{\sum_{i\geq 0}
\hat{\degree}_\tim(i)} =  \frac{\hat{\degree}_\tim(i)}{N_0} =
\g_\tim(i) + \frac{1}{N_0}\noisee_\tim(i) $$
that the vertex distribution $\g_\tim$ of the graph
$\graph_\tim$ generated according to Algorithm \ref{alg0} is measured in
noise by the administrator of the social network. That is, the
measurement is \beq \obs_\tim = \g_\tim + \noise_\tim   \eeq where
$\noise_\tim = \frac{\noisee_\tim}{N_0}$. Recall that  $N_\tim = N_0$ when $\pdup = 0$.
The normalized noisy observations from the monitoring node, $\obs_\tim$,
are used to estimate the empirical probability mass function of degree of
each node.
To estimate  a time varying PMF,  the following stochastic approximation
algorithm with constant step size, $\esa$ (where $\esa$ denotes a small
positive constant), is used to estimate the empirical probability mass
function:
\begin{equation}\label{eq15}
\hg_{\tim+1} = \hg_\tim +\esa\left(\obs_\tim - \hg_\tim\right).
\end{equation}
 Note that the stochastic approximation algorithm (\ref{eq15}) does not
assume any knowledge of the Markov-modulated dynamics of the graph. The
Markov chain assumption for the random graph dynamics is only used in our
convergence and tracking analysis. Our goal is to analyze how well the
algorithm tracks the empirical node degree of the graph. This section
studies the asymptotic behavior of the estimated degree distribution. Let
$ \E\{\bg(\mc_\tim)\}$ denote the expectation of $\bg(\mc_\tim)$ with
respect to $\sigma$-algebra, $\mathcal{G}$, generated by $\{\obs_k,\quad
k \leq \tim\}$. First, we show that the difference between
$ \E\{\bg(\mc_\tim)\}$ and $\hg_\tim$, obtained by stochastic
approximation, is bounded and the upper bound depends on $\esa$ and
$\emc$.

 \subsection{Tracking Error of the Stochastic Approximation
Algorithm}\label{subsec:bound}
 Recall that the tracking error is $\tg_\tim = \hg_\tim -
\E\{\bg(\mc_\tim)\}$. Theorem~\ref{theo2} below shows that the difference
between sample path and the expected probability mass function is small -
-implying that the stochastic approximation algorithm can successfully
track the Markov-modulated node distribution given noisy measurements (We
again emphasize that not knowledge of the Markov chain parameters are
required in the algorithm). It also finds the order of this difference in
terms of $\esa$ and $\emc$.

\begin{Theorem}
\label{theo2}
Consider the random graph $(M, \A^{\emc},\pi_0,\p,\pdel,\pdup,
\graph_0)$. Suppose that $\emc^2 = o(\esa)$\footnote{Note that in this paper, we assume that $\emc = O(\esa)$, therefore $\emc^2 = o(\esa)$ is a consequence.}.
Then for sufficiently large $\tim$ the tracking error of the stochastic
approximation (\ref{eq:sa}) is
\begin{equation}
\mathbf{E}|\tg_\tim|^2 = O\left(\esa+\emc +\frac{\emc^2}{\esa}\right).
\end{equation}
\qed
\end{Theorem}

The proof of Theorem~\ref{theo2} is presented in Appendix~\ref{ap:bound}.
In the proof, the perturbed Liapunov function methods are used.
As a corollary of  Theorem \ref{theo2}, we obtain the following
mean squares convergence result.
\begin{Corollary} Under the conditions of Theorem \ref{theo2}, if $\emc = O(\esa)$
we have
$$ \mathbf{E} |\tg_\tim|^2 = O(\esa).$$
and therefore,
$$\limsup_{\e\to 0} \mathbf{E} |\tg_\tim|^2 = 0.$$
\end{Corollary}

\subsection{Limit System of Regime-Switching Ordinary Differential
Equations}\label{subsec:ode}
Theorem~\ref{theo3} shows that the sequence of estimates generated by the
stochastic approximation algorithm (\ref{eq15}) converges weakly to the
dynamics of a Markov-modulated ordinary differential equation.

\begin{Theorem}\label{theo3}
Consider the Markov-modulated random graph generated by
Algorithm~\ref{alg0}, and the sequence of estimates
$\{\hg_\tim\}$ generated by stochastic approximation algorithm
(\ref{eq15}). Assume condition (A) holds, and $\emc = O(\esa)$. Define
the continuous-time interpolated process
\begin{equation}
\hg^\esa(\ct) = \hg_\tim, \  \mc^\esa(\ct) = \mc_\tim
 \ \ct\in[\tim\esa,(\tim+1)\esa).
\end{equation}
 Then as $\esa\rightarrow 0$,
$(\hg^\esa(\cdot),\mc^\esa(\cdot))$ converges weakly to
$(\hg(\cdot),\mc(\cdot))$ such that $\mc$ is continuous-time Markov chain
with generator $Q$ and $\hg(\cdot)$ satisfies the Markov-modulated
ordinary differential equation (ODE)
 \begin{equation}
 \label{diff2}
 \frac{d\hg(\ct)}{d\ct} = -\hg(\ct) + \bg(\mc(\ct)), \quad\hg(0) = \hg_0,
 \end{equation}
 where $\bg(\mc)$ is defined in (\ref{eq:gbar}).
 \qed
\end{Theorem}

Note that (\ref{diff2}) is a Markov-modulated ordinary differential
equation. The above theorem asserts that the empirical measure obtained by
stochastic approximation algorithm (\ref{eq15}) converges weakly to
Markovian switched ODE (\ref{diff2}).  As mentioned in
Sec.\ref{sec:intro}, this is unusual since typically in averaging
 of stochastic approximation algorithms, convergence occurs to a
deterministic differentia equation. The intuition behind that the
estimates obtained by (\ref{eq15}) converges to a Markov-modulated ODE
(rather than a deterministic ODE) is that the Markov chain (with
transition matrix $I + \emc Q$ ) evolves on the same time scale as the
stochastic approximation algorithm with step size $\esa$ (when $\emc =
O(\esa)$). If the Markov chain evolved on a faster time scale, then the
limiting dynamics would indeed be a deterministic ODE weighed by the
stationary distribution for the Markov chain. If the Markov chain evolved
slower than the dynamics of the stochastic approximation algorithm, then
the asymptotic behavior would also be a deterministic ODE with the Markov
chain being a constant.

\subsection{Scaled Tracking Error}\label{subsec:er}
The following theorem studies the behavior of the scaled tracking error
between the estimates generated by the stochastic approximation algorithm
(\ref{eq15}) and the expected degree distribution and proves that this
error should also satisfy a switching diffusion equation.
Theorem~\ref{theo4} gives a functional central limit theorem for this
scaled tracking error. Let $\ser_k = \frac{\hg_k -
\mathbf{E}\{\bg(\mc_k)\}}{\sqrt{\esa}}$ denote the scaled tracking error.

\begin{Theorem}\label{theo4}
Assume condition (A) holds.
% the interpolated sequence of iterates,
Define
$\ser^\esa(t) = \ser_k$ for $t \in [k\esa, (k+1)\esa)$.
Then $(\nu^\e\cd,\theta^\e\cd)$ converges
weakly $(\nu\cd,\theta\cd)$ such that
$\nu\cd$ is the solution of the following Markovian switched diffusion
process
\begin{equation}\label{de}
\ser(t) = -\int_0^t \ser(s) ds + \int_0^t
\cov^{\frac{1}{2}}(\mc(\tau))d\omega(\tau),
\end{equation}
where $\omega(\cdot)$ is an $\rr^{N_0}$-dimensional standard Brownian motion.
% and the last integral is interpreted as an Ito integral.
The covariance matrix,
$\cov(\mc)$, in (\ref{de}) can be explicitly computed as
\beq\label{eq:cov3}
\cov(\mc) = Z(\mc)'D(\mc) + D(\mc)Z(\mc) - D(\mc) - \bg(\mc)\bg'(\mc).
\eeq
Here, $D(\mc) = \diag(\bg(\mc,1),\ldots,\bg(\mc, \s))$ and $Z(\mc) =
\left(I - \ttrue(\mc) +\mathbf{1}\bg'(\mc)\right)^{-1}$ where
$\ttrue(\mc_\tim)$ and $\bg(\mc)$ are defined in (\ref{eq:B})
and (\ref{eq:gbar}), respectively. \qed
\end{Theorem}

For general switching processes, we refer to \cite{YinZ10}. In fact, more
complex continuous-state dependent switching rather than Markovian switching was considered there.
 Eq. (\ref{eq:cov3}) reveals that the covariance matrix of the tracking
error depends on $\ttrue(\mc)$ and $\bg(\mc)$ and consequently
on the parameters of $\p$ and $\pdel$ of the random graph. Recall from Sec.\ref{sec:markov} that
$\ttrue(\mc)$ is the transition matrix of the Markov chain which
models the evolution of the expected degree distribution in Markov
modulated random graph and can be computed from Theorem~\ref{theo:mu}.
%The covariance of the tracking error which can be explicitly computed from (\ref{eq:cov3}), is useful for computational purposes.
 We can interpret the covariance matrix in terms of searchability of the graph defined in Sec.\ref{sec:markov}.
Sec.\ref{sec:num} provides numerical examples that show that the trace of
the covariance matrix $\cov(\mc)$ is proportional to the searchability of
the graph generated by Algorithm~\ref{alg0}. Numerical examples in
Sec.\ref{sec:num} also show that the trace of covariance of the tracking
error is proportional to the average degree of nodes.
\section{Discussion and Extension: Power Law Component for Infinite Duplication-deletion Random Graph Without Markovian Dynamics}\label{subsec:pl}
   In Sec.\ref{sec:markov}, a degree distribution analysis is provided for the fixed size Markov-modulated random graph generated according to Algorithm~\ref{alg0} with $\pdup = 0$. This section extends the results of Sec.\ref{sec:markov} to the \textit{infinite duplication-deletion random graph without Markovian dynamics}. Here, we investigate the random graph generated according to Algorithm
\ref{alg0} with $\pdup = 1$ and when there are no Markovian dynamics,
that is, $M = 1$. Since $\pdup = 1$ for  $\tim \geq 0$,
$\graph_{\tim+1}$ has one more vertex compared to  $\graph_\tim$. In
particular, since $\graph_0$ is an empty set, $\graph_\tim$ has
$\tim$ nodes, that is, $N_\tim = \tim$. In this section, employing the same approach used in the proof of Theorem~\ref{theo:mu}, it is shown that the infinite duplication-deletion random graph without Markovian dynamics generated by Algorithm \ref{alg0} with $\pdup = 1$ satisfies a power law and an expression
 is derived for the power law component.
 Let us first define the power law:
\begin{Definition}[Power Law]
\label{def1}
Consider the infinite duplication-deletion random graph without Markovian dynamics
 generated according to
Algorithm \ref{alg0} with 7-tuple $(M,A^{\emc},\pi_0,\p,\pdel,\pdup,\graph_0)$. Let $n_k$ denote the number of nodes of degree
$k$ in a random graph $\graph_\tim$. Then $\graph_\tim$ satisfies a
power law distribution if $n_k$ is proportional to $k^{-\beta}$ for a
fixed $\beta > 1$ : $\log n_k = \alpha -\beta \log k$,  where $\alpha$ is
a constant. $\beta$ is called power law component.
\end{Definition}
\begin{Theorem}
\label{theo:pl}
    Consider the infinite random graph with Markovian dynamics $\graph_n$ obtained by Algorithm \ref{alg0}
with 7-tuple $(1,1,1,1,p,q,\graph_0)$ with the expected degree distribution $\bg_n$. As $\tim \rightarrow \infty$,
$G_n$ satisfies  a power law. That is \begin{equation} \log \bg_\tim(i) = \alpha -\beta \log i,
\end{equation} where the power law
component, $\beta$, can be computed from following equation.
    \begin{align}\label{eq:pl}
(1+\pdel)( \p^{\beta - 1 }+\p\beta - \p ) =1 + \beta \pdel,
\end{align}
where $\p$ and $\pdel$ are the probabilities defined in duplication and
deletion steps.\qed
\end{Theorem}
\noindent{\em Remark 1. Outline of Proof}:
The proof of Theorem~\ref{theo:pl}, which is presented in
Appendix~\ref{ap:pl}, consists of two steps: (i) finding the power law
component and (ii) showing that the degree distribution converges to a
power law as $\tim \rightarrow \infty$. To find the power law component,
we derive a recursive equation for the number of nodes with degree
$i+1$ at time $\tim + 1 $, $\degree_{\tim+1}(i+1)$, in terms of degree of
nodes in graph $\graph_\tim$. Then, this recursive equation is rearranged
to equation for the power law component. To prove that the degree
distribution satisfies a power law, we define a new parameter $h_\tim(i)
= \frac{1}{\tim}\sum_{k=1}^{i}\E\{\degree_{\tim}(k)\}$ and we show that
$\lim_{\tim \rightarrow \infty} h_\tim(i) = \sum_{k=1}^iCk^{-
\beta}$ where $\beta$ is the power law component computed by the solving
the recursive equation.
Theorem~\ref{theo:pl} asserts that the infinite duplication-deletion random graph without Markovian dynamics generated by Algorithm~\ref{alg0} satisfies a power
law and provides an expression for the power law component. The
significance of this theorem is that it ensures that with use of one
single parameter (the power law component), we can describe the degree
distribution of large numbers of nodes in graphs that model social
networks.

\begin{figure}[h]
\centering
\hspace{-1.2cm}\scalebox{.9}{\includegraphics{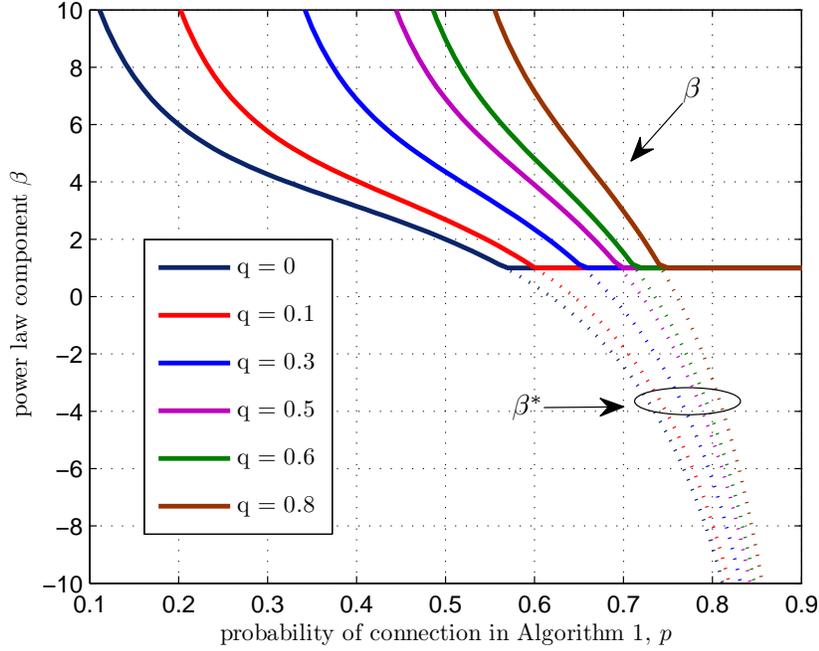}}
\caption{The power law component for the non-Markovian random graph
generated according to Algorithm~\ref{alg0} obtained by (\ref{eq:pl}) for
different values of $\p$ and $\pdel$  in Algorithm~\ref{alg0}. }
% http://www.biomedcentral.com/1471-2105/5/147
\label{beta}
\end{figure}

\noindent{\em Remark 2. Power Law Component}: Let $\beta^*$ denote the
solution of (\ref{eq:pl}). Then the power law component is defined as  $\beta =\max
\{1,\beta^*\}$. Fig.\ref{beta} shows the the power law component and
$\beta^*$  versus $\p$ for different values of probability of deletion,
$\pdel$. As can be seen in Fig.\ref{beta}, the power law component is
increasing in $\pdel$ and decreasing in $\p$.

\section{Numerical Examples}\label{sec:num}
In this section, numerical examples are given to illustrate the results
from  Sec.\ref{sec:markov}, Sec.\ref{sec:pmf}, and Sec.\ref{subsec:pl}. \\
The main conclusions are:\begin{enumerate}[(i)]
\item The infinite duplication-deletion random graph without Markovian dynamics generated by Algorithm~\ref{alg0} satisfies
a power law as stated in Theorem~\ref{theo:pl}. This is illustrated in Example 1 below.
\item The degree distribution of the fixed size duplication-deletion random graph generated by Algorithm~\ref{alg0} can be computed from Theorem~\ref{theo:mu}. When $N_0$ (the size of the random graph) is sufficiently large, numerical results show that the degree distribution satisfies a power law as well. This is shown in Example 2 below.
\item The estimates obtained by stochastic approximation algorithm
(\ref{eq15}) follow the expected probability distribution precisely
without information about the Markovian dynamics. This is illustrated in
Example 3 below.
\item The larger the trace of the asymptotic covariance of the scaled
tracking error, the greater the average degree of nodes and the
searchability of the graph. This is illustrated in Example 4 below.
    \end{enumerate}

 {\em Example 1}: Consider an infinite duplication-deletion random graph without Markovian dynamics generated by Algorithm~\ref{alg0} with $\p = 0.5$ and $\pdel = 0.1$. Theorem~\ref{theo:pl} implies that the degree sequence of the resulting graph satisfies a power law with exponent computed using (\ref{rec}). Fig.\ref{sim2} shows the number of nodes with specific degree on a logarithmic scale for both horizontal and vertical axes. It can be inferred from the linearity in  Fig.\ref{sim2} (excluding the nodes with very small degree), that the resulting graph from duplication-deletion process satisfies a power law. As can be seen in Fig.\ref{sim2}, the power law is a better approximation for the middle points compared to both ends.
\begin{figure}[h]
\centering
\hspace{-1.2cm}\scalebox{.8}{\includegraphics{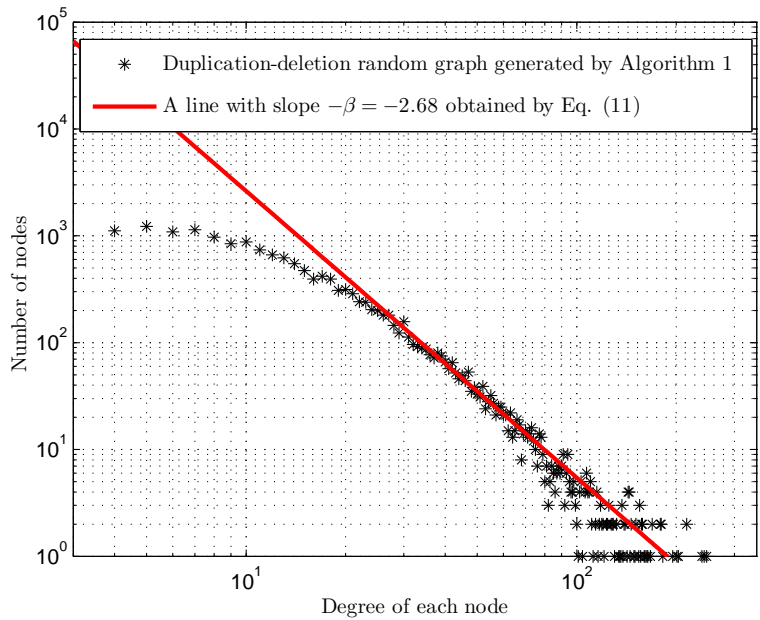}}
\caption{Illustration of Theorem~\ref{theo:pl}: The degree distribution
of the duplication-deletion random graph satisfies a power law. The
parameters are specified in Example 1 of Sec.\ref{sec:num}. }
% http://www.biomedcentral.com/1471-2105/5/147
\label{sim2}
\end{figure}

{\em Example 2}: Consider the fixed size duplication-deletion random graph obtained by Algorithm~\ref{alg0} with $\pdup = 0$, $N_0 = 10$, $p = 0.4$, and $q = 0.1$. (We consider no Markovian dynamics here to illustrate Theorem~\ref{theo:mu}.) Fig.~\ref{mc} depicts the degree distribution of the fixed size duplication-deletion random graph obtained by Theorem~\ref{theo:mu}. As can be seen in Fig.~\ref{mc}, the computed degree distribution is close to that obtained by simulation. The numerical results show that the degree distribution of the fixed size random graph also satisfies a power law for some values of $p$ when the size of random graph is sufficiently large. Fig.~\ref{mc2} shows the number of nodes with specific degree for the fixed size random graph obtained by Algorithm~\ref{alg0} with $\pdup = 0$, $N_0 = 1000$, $p = 0.4$, and $q = 0.1$ on a logarithmic scale for both horizontal and vertical axes.

\begin{figure}[h]
\begin{minipage}[b]{.5\textwidth}
\centering
\hspace{-1.2cm}\scalebox{.6}{\includegraphics{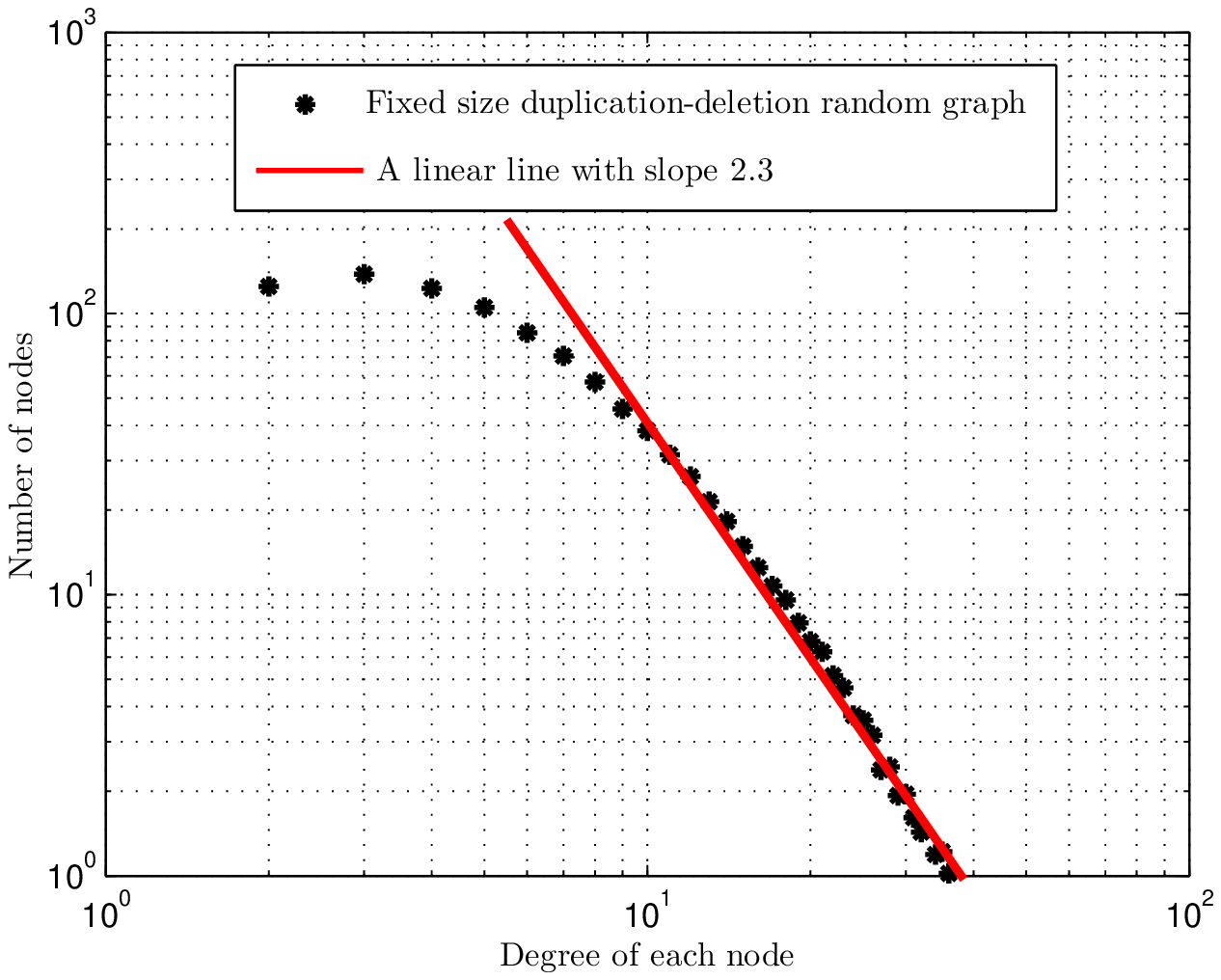}}
\caption{ The degree distribution
of the fixed size duplication-deletion random graph satisfies a power law when $N_0$ is sufficiently large. The
parameters are specified in Example 2 of Sec.\ref{sec:num}. }
% http://www.biomedcentral.com/1471-2105/5/147
\label{mc2}
\end{minipage}
\hspace{.5cm}
\begin{minipage}[b]{.5\textwidth}
\centering
\hspace{-1.2cm}\scalebox{.6}{\includegraphics{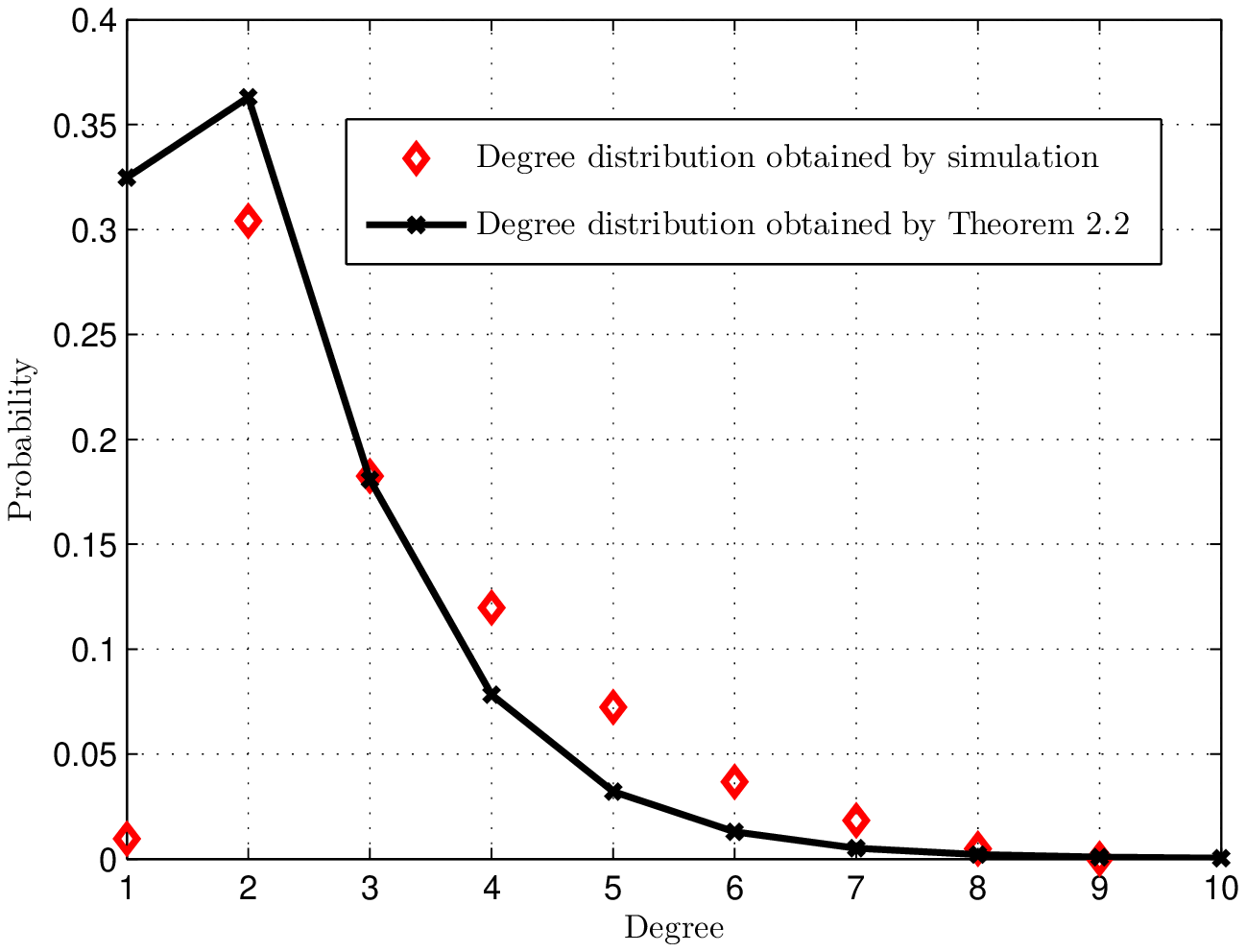}}
\caption{Illustration of Theorem~\ref{theo:mu}: The degree distribution
of the fixed size duplication-deletion random graph. The
parameters are specified in Example 2 of Sec.\ref{sec:num}. }
% http://www.biomedcentral.com/1471-2105/5/147
\label{mc}
\end{minipage}
\end{figure}

 {\em Example 3}: Consider the fixed size Markov-modulated duplication-deletion random graph
generated by Algorithm \ref{alg0} with $\pdup = 0$ and $N_0 = 500$. Assume that the underlying Markov chain has three states,
$M=3$. We choose the following values for probabilities of connection and
deletion: state (1): $p = q = 0.05$, state (2): $p = 0.2$ and $q = 0.1$,
and state (3): $p = 0.4$, $q = 0.15$.  The sample path of the Markov chain jumps at times $n=3000$ from state (1) to state (2) and $n=6000$ from state (2) to state (3). As the state of
the Markov chain changes, the expected degree distribution, $\bg(\mc)$,
obtained by (\ref{eq:gbar}) evolves over time. The corresponding values
for the expected degree distribution (for $i=3$) are shown in
Fig.\ref{samplepath} by a dotted line. The estimated probability mass
function, $\hg_\tim$, obtained by the stochastic approximation algorithm (\ref{eq:sa}) is plotted in
Fig.\ref{samplepath} using a solid line. The figure shows that the
estimates obtained by the stochastic approximation algorithm (\ref{eq15})
follow the expected degree distribution obtained by (\ref{eq:gbar})
precisely without any information about the Markovian dynamics.

 {\em Example 4}: Consider the fixed size Markov-modulated duplication-deletion random graph
obtained by Algorithm \ref{alg0} with $M = 91$ and $\pdup =0$ and $N_0 = 1000$. For each value of $\p(\mc) = 0.04 + \mc \times 0.01,
\mc \in \{1,2,\ldots, 91\}$ and $\pdel \in\{0.05, 0.1, 0.15, 0.2\}$, we
compute $\transition(\mc)$ from (\ref{eq:L}) and consequently the
stationary distribution, $\bg(\mc)$, from (\ref{eq:gbar}).  As expected,
the stationary distribution does not depend on $\pdel$ because only the deletion step in Algorithm \ref{alg0} occurs
with probability $\pdel$. From $\bg(\mc)$, we  compute the average degree
of nodes, $\de$. Fig.\ref{cov3} shows the average degree of nodes versus
the probability of the connection in Algorithm~\ref{alg0}. As can be seen
in Fig.\ref{cov3}, with increasing the probability of connection in
Algorithm \ref{alg0}, the average degree of nodes in the graph (which is
a measure for the connectivity of the graph, see \cite{complex})
increases.
\begin{figure}[h]
\begin{minipage}[b]{.5\textwidth}
\centering
\hspace{-1.2cm}\scalebox{.6}{\includegraphics{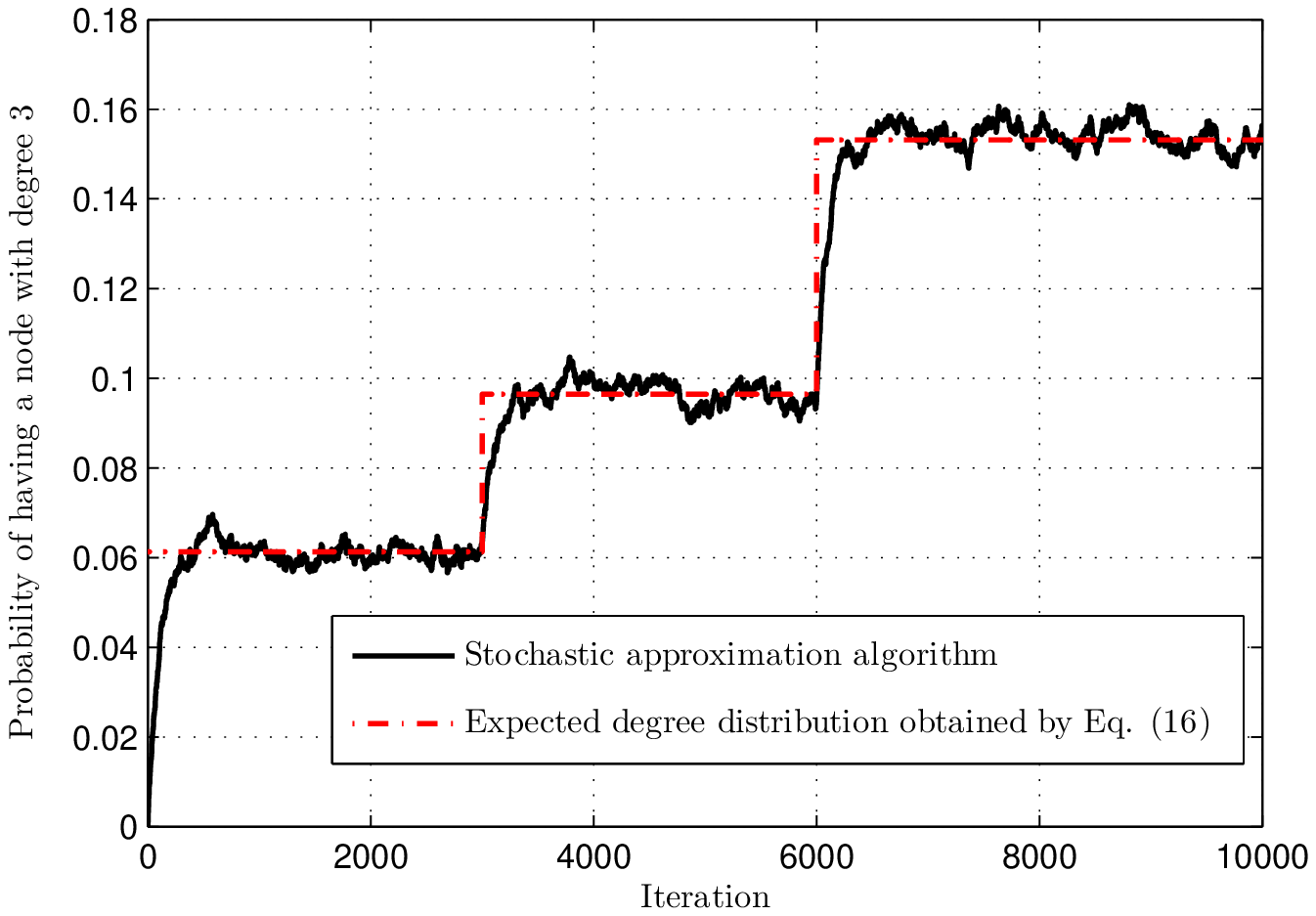}}
\caption{Illustration of Theorem~\ref{theo2}: The estimated probability
mass function obtained by the stochastic approximation algorithm
(\ref{eq15}) follows the expected probability distribution precisely
without information about the Markovian dynamics. The parameters are
specified in Example 4 of Sec.\ref{sec:num}.}
% http://www.biomedcentral.com/1471-2105/5/147
\label{samplepath}
\end{minipage}
\hspace{.5cm}
\begin{minipage}[b]{.5\textwidth}
\centering
\hspace{-1.2cm}\scalebox{.6}{\includegraphics{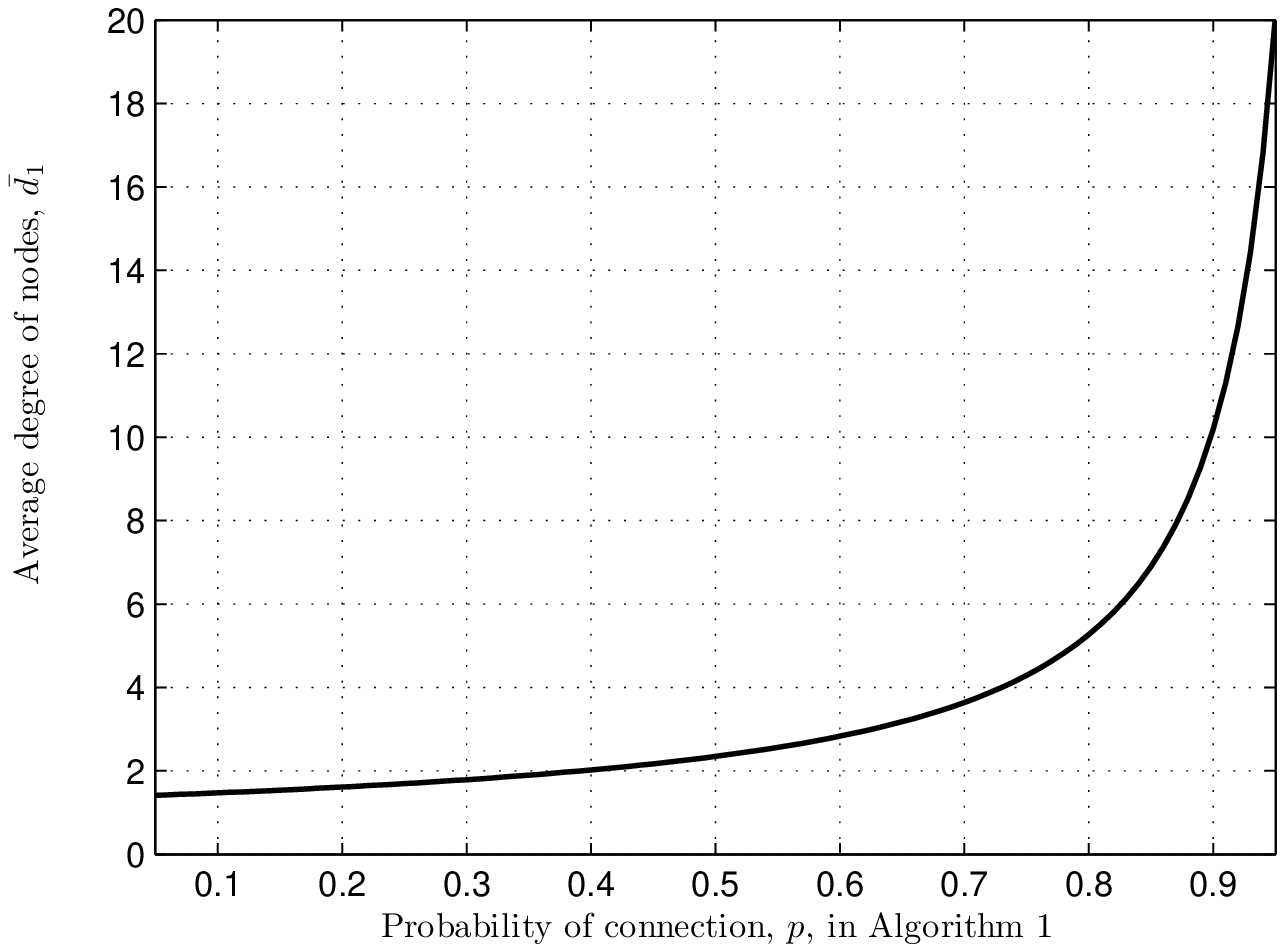}}
\caption{The average degree of nodes (as a measure of connectivity) of
the fixed size Markov-modulated duplication-deletion random graph obtained by Algorithm
\ref{alg0} for different values of the probability of connection, $\p$,
in Algorithm \ref{alg0}. The parameters are specified in Example 4 of
Sec.\ref{sec:num}. }
% http://www.biomedcentral.com/1471-2105/5/147
\label{cov3}
\end{minipage}
\end{figure}

Then for  each value of $\p(\mc) = 0.04 + \mc \times 0.01,  \mc \in
\{1,2,\ldots, 91\}$ and $\pdel \in\{0.05, 0.1, 0.15, 0.2\}$, the
covariance matrix is computed using (\ref{cov3}). Fig.\ref{cov1} depicts
the trace of the covariance matrix, $\trace\left(\cov(\mc)\right)$, for
each value of $\p$ and $\pdel$ versus the corresponding average degree of
nodes (for each value of $\p$). As can be seen in Fig.\ref{cov1}, the
trace of the covariance matrix is larger when the average degree of nodes
is higher (the graph is highly connected).

Recall from Lemma \ref{lem:search}, the order of delay in the searching
problem can be computed by $\delay(N_0) = O\left(\frac{N_0\de}{\dd
- \de}\right)$. Knowing the degree distribution $\bg(\mc)$, $\de$ and
$\dd$ can be computed for each value of $\p \in \{0.05, 0.06,\ldots,
0.95\}$. Fig.\ref{cov2} shows the trace of the covariance matrix versus
$\left(\frac{\de}{\dd-\de}\right)$  as a measure of the searchability for
each value of $\pdel \in\{0.05, 0.1, 0.15, 0.2\}$. As can be seen in
Fig.\ref{cov2}, the trace of covariance matrix is larger when the order
of delay in the search problem in (\ref{eq:delay}) is smaller\footnote{This means that the target node can be found in the search problem with smaller number of steps.}.

\begin{figure}[h]
\begin{minipage}[b]{.5\textwidth}
\centering
\hspace{-1.2cm}\scalebox{.62}{\includegraphics{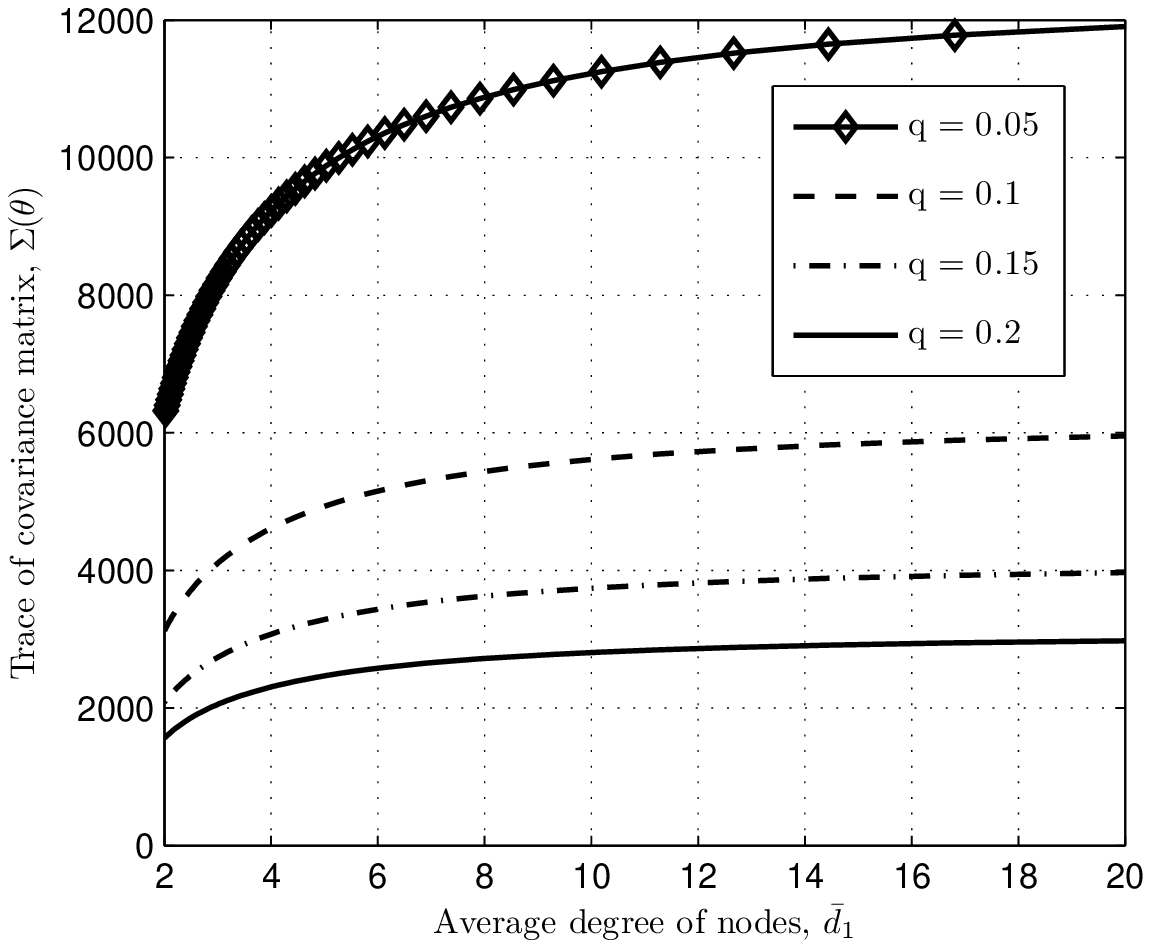}}
\caption{The trace of the covariance matrix of the scaled tracking error,
$\trace\left(\cov(\mc)\right)$, versus the average degree of nodes as a
measure of connectivity of the network. The parameters are specified in
Example 3 of Sec.\ref{sec:num}. }
% http://www.biomedcentral.com/1471-2105/5/147
\label{cov1}
\end{minipage}
\hspace{.5cm}
\begin{minipage}[b]{.5\textwidth}
\centering
\hspace{-1.2cm}\scalebox{.62}{\includegraphics{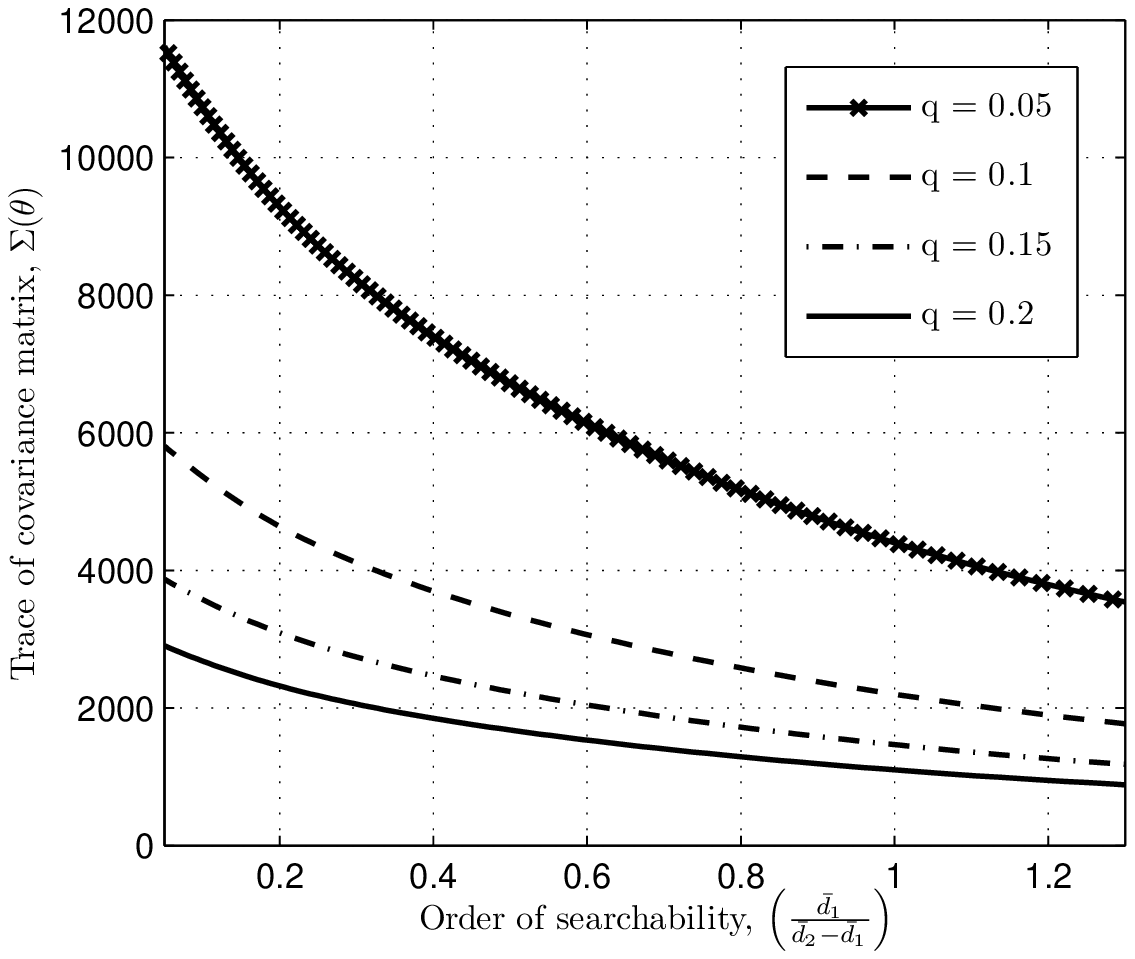}}
\caption{The trace of the covariance matrix of the scaled tracking error,
$\trace\left(\cov(\mc)\right)$, versus the order of delay in the
searching problem as a measure of searchability of the network. The
parameters are specified in Example 3 of Sec.\ref{sec:num}.  }
% http://www.biomedcentral.com/1471-2105/5/147
\label{cov2}
\end{minipage}
\end{figure}
\section{Conclusion}\label{sect7}
This paper analyzed the dynamics of a duplication-deletion graph where at
each time instant, one node can either join or leave the graph (An
extension to the duplication model of~\cite{duplication,bebek25}). The
power law component for such graph was computed using the result of
Theorem~\ref{theo:pl}. Also a Markov-modulated random graph was proposed
to model the social networks whose evolution changes over time. Using the
stochastic approximation algorithms, the probability mass function of
degree of each node is estimated. Then, an upper bound was derived for
the distance between the estimated and the expected PMF. As a result of
this bound, we showed that the scaled tracking error between the expected
PMF and the estimated one weakly converges to a diffusion process. From
that, the covariance of this error can be computed.  Finally, we
presented a discussion on application of this work in controlling a
social network using the degree distribution obtained by stochastic
approximation. In this case it is assumed that the network manager
observes the degree of active users and this observation is noisy due to
the activity profile of users. Using the estimated degree distribution,
the network manager can track the level of connectivity (by computing the
orders of size of giant component) and the searchability of the network
(by computing the order of delay).

\appendix
\subsection{Proof of Theorem~\ref{theo:mu}}\label{ap:mu}
\begin{proof}
To find the degree distribution of nodes, we find a relation between the number of nodes
with specific degree at time $\tim$ and the degree distribution of the
graph at time $\tim-1$. Given the resulting graph at time $\tim$, we are
trying to find the expected number of nodes with degree $i+1$ at time
$\tim+1$. The following events can occur that result in a node with
degree $i+1$ at time $\tim+1$:
\begin{itemize}
\item{A node with degree $i$ is chosen at the duplication step as a
parent node. In this case, there will be another edge connecting the new
node to the parent node in the edge-duplication step. Probability of
choosing a node with degree $i$ is $\frac{\degree_\tim(i)}{N_\tim}$. If a
node with degree $i$ is not chosen itself but one of its neighbors is
selected as parent node, there is also a chance for this node to have
another edge (with probability of $\p$). This node has $i$ neighbors
therefore, the corresponding probability is $\p.i$. So the probability
that the degree of such node increases by $1$ in the duplication step is
$\frac{1+pi}{N_\tim}$. Also, in the deletion step nor this node neither any
its neighbors should be selected in the edge-deletion step. With the same
discussion, the associated probability is $\left(1-
\frac{\pdel(i+1)}{N_\tim}\right)$. If deletion step occurs, another node is
generated and connected to the graph as described in deletion-step in
Sec.\ref{sec:intro}. Nor this node (the node with degree $i$) and none of
its neighbors should be selected in this step. So the probability that
this node remains unchanged after deletion step is:  $\left(1-
\frac{\pdel(i+1)+\pdel (1+pi)}{N_\tim}\right)$ }
\item{A node with degree $i+1$ at time $\tim$ does not change during
duplication and deletion processes. To be unchanged in both duplication
and deletion steps, this node or any of its neighbors should not be
chosen in both duplication and deletion steps. The probability of being
unchanged during these processes for an specific node can be computed
from $\left(1-
\frac{\pdel(i+2)+\pdel\big(1+\p(i+1)\big)}{N_\tim}\right)\left(1-
\frac{\p(i+1)+1}{N_\tim}\right)$} and total number of such nodes at time
$\tim$ is $f_\tim(i+1)$.
    \item{The degree of the most recently generated node (in the vertex-
duplication) increases to $i+1$ in the edge-duplication step. This means
that, this node is connected to $``i"$ neighbors of the parent node and
remains unchanged in the deletion step. The probability of this scenario
is
        $\\ \left(1-\frac{\pdel(i+2) +
\pdel\big(1+\p(i+1)\big) }{N_\tim}\right)\sum_{j\geq
i}\frac{1}{N_\tim}f_\tim(j){{j}\choose{i}}\p^i(1-\p)^{j-i}$.}
\item{A node with degree $i+2$ remains unchanged in the duplication step
and one of its neighbors is eliminated in the deletion step. The
probability of this event is  $\pdel\left(\frac{i+2}{N_\tim}\right)\left(1-
\frac{\p(i+2)+1}{N_\tim}\right)$.}
\item{The degree of the node generated in the deletion-step increases to
$i+1$ (As described in Sec.\ref{sec:intro}, in deletion-step to maintain
the total number of nodes, a new node is generated and connected to the
graph). The probability of this scenario is $ \pdel\sum_{j\geq
i}\frac{1}{N_\tim}f_\tim(j){{j}\choose{i}}\p^i(1-\p)^{j-i}$.}
\item{A node with degree $i$ remains unchanged in the duplication step
and the same node or one of its neighbors selected in the duplication
part of the deletion step. The corresponding probability is
$\frac{\pdel(1 + pi)}{N_\tim}\left(1 - \frac{1+pi}{N_\tim}\right)$ }
\item{The degree of a node with $i+1$ neighbors increases in the
duplication step and one of its neighbors is eliminated in the deletion
step. The corresponding probability is
$\pdel\left(\frac{i+2}{N_\tim}\right)\left(\frac{\p((i+1)+1)}{N_\tim}\right)$
.}
\end{itemize}
 Let $\Omega$ denote the set of all arbitrary graphs and
$\mathcal{F}_\tim$ denote the sigma algebra generated by graphs
$\graph_\tau, \tau \leq \tim$. Considering the above events that result
in a node with degree $i+1$ at time $\tim+1$, the following recurrence
formula can be derived for the conditional expectation of
$f_{\tim+1}(i+1)$:
\begin{align}
\label{eq1}
\mathbf{E}\{f_{\tim+1}(i+1)|\mathcal{F}_\tim\}&=\left(1-
\pdel\frac{(i+2)+ (1+ \p(i+1))}{N_\tim}\right)\left(1-
\frac{\p(i+1)+1}{N_\tim}\right)f_\tim(i+1)\nonumber\\
&\ +\left(1-
\pdel\frac{(i+1) + (1
+pi)}{N_\tim}\right)\left(\frac{1+pi}{N_\tim}\right)f_\tim(i)\nonumber\\
& \ +\left(1-\pdel\frac{(i+2) + (1 +\p(i+1))}{N_\tim}\right)\sum_{j\geq
i}\frac{1}{N_\tim}f_\tim(j){{j}\choose{i}}\p^i(1-\p)^{j-i}\nonumber\\
& \ +\pdel\sum_{j\geq i}\frac{1}{N_\tim}f_\tim(j){{j}\choose{i}}\p^i(1-\p)^{j-
i}\nonumber\\
& \ +\pdel\left(\frac{i+2}{N_\tim}\right)\left(1-\frac{\p(i+2)+1}{N_\tim}\right)
f_\tim(i+2)\nonumber\\
& \ +\frac{\pdel(1 + pi)}{N_\tim}\left(1 -
\frac{1+pi}{N_\tim}\right)f_\tim(i)\nonumber\\
& \ +\pdel\left(\frac{i+2}{N_\tim}\right)\left(\frac{\p((i+1)+1)}{N_\tim}\right)
f_\tim(i+1).
\end{align}
 Let $\baf^\mc_\tim(i)= \mathbf{E}\{f_\tim(i)|\mc_\tim = \mc\}$. By taking expectation of
both sides of (\ref{eq1}) with respect to trivial sigma algebra
$\{\Omega,\emptyset\}$, the smoothing property of conditional
expectations yields.
\begin{align}
\label{eq:1}
\baf^\mc_{\tim+1}(i+1)&=\left(1- \pdel\frac{(i+2)+ (1+
\p(i+1))}{N_\tim}\right)\left(1-
\frac{\p(i+1)+1}{N_\tim}\right)\baf^\mc_\tim(i+1)\nonumber\\\
& \ + \left(1-
\pdel\frac{(i+1) + (1
+pi)}{N_\tim}\right)\left(\frac{1+pi}{N_\tim}\right)\baf^\mc_\tim(i)\nonumber\\
&\ +\left(1-\pdel\frac{(i+2) + (1 +\p(i+1))}{N_\tim}\right)\sum_{j\geq
i}\frac{1}{N_\tim}\baf^\mc_\tim(j){{j}\choose{i}}\p^i(1-\p)^{j-i}\nonumber\\
&\  +\pdel\sum_{j\geq i}\frac{1}{N_\tim}\baf^\mc_\tim(j){{j}\choose{i}}\p^i(1-
\p)^{j-i}\nonumber\\
 &\ +\pdel\left(\frac{i+2}{N_\tim}\right)\left(1-\frac{\p(i+2)+1}{N_\tim}\right)
\baf^\mc_\tim(i+2)\nonumber\\
&\ +\frac{\pdel(1 + pi)}{N_\tim}\left(1 -
\frac{1+pi}{N_\tim}\right)\baf^\mc_\tim(i)\nonumber\\
&\ +\pdel\left(\frac{i+2}{N_\tim}\right)\left(\frac{\p((i+1)+1)}{N_\tim}\right)
\baf^\mc_\tim(i+1).
\end{align}
Assuming that size of the graph is sufficiently large,  each term like
$\frac{\baf_\tim(i')}{N_\tim^2}$ can be neglected for large $N_\tim$. So (\ref{eq:1}) can be re-written as
\begin{align}
\label{eq:barf}
\baf^{\mc}_{\tim+1}(i+1)& = \left(1- \frac{\pdel(\mc)(i+2)+
\pdel(\mc)\big(\p(\mc)(i+1)+1\big)}{N_\tim}\right)\baf^{\mc}_\tim(i+1)
\nonumber\\
&\ +\left(\frac{(1+\p(\mc)
i)\pdel(\mc)}{N_\tim}\right)\baf^{\mc}_\tim(i)+\pdel(\mc)\left(\frac{i+2}
{N_\tim}\right)\baf^{\mc}_\tim(i+2)\nonumber\\
&\ + \pdel(\mc)\sum_{j\geq
i}\frac{1}{N_\tim}\baf^{\mc}_\tim(\mc,j){{j}\choose{i}}\p(\mc_{\tim+1})^i
(1-\p(\mc_{\tim+1}))^{j-i}.
\end{align}

 Using (\ref{eq:1}), we can write the following recursion for the
$(i+1)$-th element of $\bg^{\mc}({\tim+1})$.
\begin{align}
\label{eq:barg}
\bg^{\mc}_{\tim+1}(i+1)& = \left(\frac{N_\tim-\left(
\pdel(\mc)(i+2)+
\pdel(\mc)\big(\p(\mc)(i+1)+1\big)\right)}{N_{\tim+1}}
\right)\bg^{\mc}_\tim(i+1)\nonumber\\
&\ +\left(\frac{(1+\p(\mc)
i)\pdel(\mc)}{N_{\tim+1}}\right)\bg^{\mc}_\tim(i)+\pdel(\mc)\left(\frac{i
+2}{N_{\tim+1}}\right)\bg^{\mc}_\tim(i+2)\nonumber\\
&\ +
\pdel(\mc)\sum_{j\geq
i}\frac{1}{N_{\tim+1}}\bg^{\mc}_\tim(j){{j}\choose{i}}\p(\mc)^i(1-
\p(\mc))^{j-i}.
\end{align}

Since the probability of duplication step $\pdup = 0$, the number of vertices does not increase. Thus, $N_n = N_0$ and (\ref{eq:barg}) can be written as
\begin{align}
\label{eq:barg2}
\bg^{\mc}_{\tim+1}(i+1)=& \Big(1 - \frac{1}{N_0}\left(\pdel(\mc)(i+2)+
\pdel(\mc)\big(\p(\mc)(i+1)+1\big)\right)\Big)\bg^{\mc}_\tim(i+1)
\nonumber\\
&+\frac{1}{N_0}\Big((1+\p(\mc)
i)\pdel(\mc)\bg^{\mc}_\tim(i)+\frac{1}{N_0}\pdel(\mc)(i+2)\bg^{\mc}_\tim(i+2)
\Big)\nonumber\\&+\frac{1}{N_0} \pdel(\mc)\sum_{j\geq
i}\bg^{\mc}_\tim(j){{j}\choose{i}}\p(\mc)^i(1-\p(\mc))^{j-i}.
\end{align}

From (\ref{eq:barg2}), it is clear that the vector
$\bg^{\mc}(\mc_{\tim+1})$ depends on elements of  $\bg^{\mc}(\mc)$. In a
matrix notation, (\ref{eq:barg2}) can be re-arranged as
\beq
\label{eq:true}
\bg^{\mc}_{\tim+1} = (I + \frac{1}{N_0} \transition(\mc))\bg^{\mc}_\tim,
\eeq
where $\transition(\mc_\tim)$ is defined as (\ref{eq:L}).

To prove that $\transition(\mc_\tim)$ is a generator, we need to show
that $l_{ii} < 0$ and $\sum_{i = 1}^{N_0}l_{ki} = 0$.
\begin{align}\label{temp}
\sum_{i = 1}^{N_0}l_{ki} &= -\left(\pdel(\mc_\tim)(k+1) +
\pdel(\mc_\tim)(1 + \p(\mc_\tim)k)\right) + (1 +
\p(\mc_\tim)k)\pdel(\mc_\tim) \nonumber\\&+ \pdel(\mc_\tim)k +
\pdel(\mc_\tim)\sum_{k\leq i-1} {{k}\choose{i-1}} \p(\mc_\tim)^{i-1} (1-
\p(\mc_\tim))^{k-i+1}\nonumber\\
& = -\pdel(\mc_\tim) + \pdel(\mc_\tim)\sum_{k\leq i-1} {{k}\choose{i-1}}
\p(\mc_\tim)^{i-1} (1-\p(\mc_\tim))^{k-i+1}.
\end{align}

Let $m = i - 1$. (\ref{temp}) can be rewritten as
\begin{align}\label{temp2}
\sum_{i = 1}^{N_0}l_{ik} =& -\pdel(\mc_\tim) +
\pdel(\mc_\tim)\sum_{m=0}^k {{k}\choose{m}} \p(\mc_\tim)^{m} (1-
\p(\mc_\tim))^{k-m}\nonumber\\
=&  -\pdel(\mc_\tim) + \pdel(\mc_\tim)(1- \p(\mc_\tim))^k\sum_{m=0}^k
{{k}\choose{m}} \left(\frac{\p(\mc_\tim)}{1-\p(\mc_\tim)}\right)^{m}
\end{align}

We know that $\sum_{m=0}^k{{k}\choose{m}} a^m = \left(1 + a\right)^k$, so
(\ref{temp2}) can be written as
\begin{align}\label{temp3}
\sum_{i = 1}^{N_0}l_{ik} & =   -\pdel(\mc_\tim) +
\pdel(\mc_\tim)(1- \p(\mc_\tim))^k\left(\frac{1}{1 -
\p(\mc_\tim)}\right)^k\nonumber\\
& = 0.
\end{align}
Also it can be shown that if $\pdel(\mc_\tim) <
\frac{\p(\mc_\tim)(1-\p(\mc_\tim))}{2 + \p(\mc_\tim)}$, then $l_{ii} < 0$.
\end{proof}

\subsection{Proof of Theorem~\ref{theo:pl}}\label{ap:pl}
\begin{proof}
To prove Theorem~\ref{theo:pl}, we first compute the power law component,
$\beta$, and then we prove that the expected degree distribution
converges to the power law distribution with component $\beta$. Let $\baf_\tim(i) = \mathbf{E}\{\degree_\tim(i)\}$. Similar to (\ref{eq1}),
$\baf_\tim(\mc_\tim,i)$ can be written as
\begin{align}
\label{eq:11}
\baf_{\tim+1}(i+1)&=\left(1- \pdel\frac{(i+2)+ (1+
\p(i+1))}{N_\tim}\right)\left(1-
\frac{\p(i+1)+1}{N_\tim}\right)\baf_\tim(i+1)\nonumber\\
&\ + \left(1-
\pdel\frac{(i+1) + (1
+pi)}{N_\tim}\right)\left(\frac{1+pi}{N_\tim}\right)\baf_\tim(i)\nonumber\\
&\ +\left(1-\pdel\frac{(i+2) + (1 +\p(i+1))}{N_\tim}\right)\sum_{j\geq
i}\frac{1}{N_\tim}\baf_\tim(j){{j}\choose{i}}\p^i(1-\p)^{j-i}\nonumber\\
& \ +\pdel\sum_{j\geq i}\frac{1}{N_\tim}\baf_\tim(j){{j}\choose{i}}\p^i(1-
\p)^{j-i}\nonumber\\
& \ + \pdel\left(\frac{i+2}{N_\tim}\right)\left(1-\frac{\p(i+2)+1}{N_\tim}\right)
\baf_\tim(i+2)\nonumber\\
& \ +\frac{\pdel(1 + pi)}{N_\tim}\left(1 -
\frac{1+pi}{N_\tim}\right)\baf_\tim(i)\nonumber\\
& \ +\pdel\left(\frac{i+2}{N_\tim}\right)\left(\frac{\p((i+1)+1)}{N_\tim}\right)
\baf_\tim(i+1).
\end{align}

To compute the power law component, we can heuristically assume that
$\baf_\tim(i) = a_it$ as $N_\tim = \tim$ goes to infinity (we will prove this
precisely later on this section). Therefore, each term like
$\frac{\baf_\tim(i')}{N_\tim^2}$ can be neglected as $\tim$ approaches
infinity. So (\ref{eq:11}) can be re-written as
\begin{align}
\label{eq3}
\baf_{\tim+1}(i+1)=& \left(1- \frac{\pdel(i+2)+(1 +
\pdel)\big(\p(i+1)+1\big)}{N_\tim}\right)\baf_\tim(i+1)
+\left(\frac{(1+pi)(1+\pdel)}{N_\tim}\right)\baf_\tim(i)\nonumber\\&
+\pdel\left(\frac{i+2}{N_\tim}\right)\baf_\tim(i+2)+(1 + \pdel)\sum_{j\geq
i}\frac{1}{N_\tim}\baf_\tim(j){{j}\choose{i}}\p^i(1-\p)^{j-i}.
\end{align}
Substituting $\baf_{\tau}(j) = a_j\tau$ and $N_\tim = \tim$ in (\ref{eq3}) yields
\begin{align}
\label{eq444}
a_{i+1}(\tim+1) =& a_{i+1}\tim -
a_{i+1}\Big(\big(1+\p(i+1)\big)(1+\pdel)+\pdel(i+2)\Big)+(1+\pdel)(1+pi)a
_{i}+\pdel(i+2)a_{i+2}\nonumber\\
&+(1+\pdel)\sum_{j\geq i}a_j{j\choose i}\p^i(1-\p)^{j-i}.
\end{align}
Taking all terms with $a_{i+1}$ to the left hand side, we have
\begin{align}
\label{eq4}
a_{i+1}\Big(1+
(1+\pdel)\big(1+\p(i+1)\big)+\pdel(i+2))\Big)=&(1+\pdel)\left((1+pi)a_{i}
+\sum_{j\geq i}a_j{j\choose i}\p^i(1-\p)^{j-
i}\right)\nonumber\\&+\pdel(i+2)a_{i+2}.
\end{align}
Dividing both sides of (\ref{eq4}) by ${a_{i}}$ yields
\begin{align}
\label{eq4-2}
\frac{a_{i+1}}{a_i}\Big(1+(1+\pdel)\big(1+\p(i+1)\big)+\pdel(i+2))\Big)=&
(1+\pdel)\left((1+pi)+\sum_{j\geq i}\frac{a_{j}}{a_i}{j\choose i}\p^i(1-
\p)^{j-i}\right)\nonumber\\&+\pdel(i+2)\frac{a_{i+2}}{a_i}.
\end{align}
Solving Equation (\ref{eq4}) for $a_i$, we can complete the proof of
Theorem~\ref{theo:pl}
The following lemma whose proof can be found in \cite{duplication} is
used to solve the recurrence relation for $a_i$.

\begin{Lemma}
\label{lem1}
\begin{equation}
    \sum_{j\geq i}\frac{a_j}{a_i}{j\choose i}\p^i(1-\p)^{j-i}= \p
^{\beta-1}+ {\rm O}\left(\frac{1}{i}\right).
    \end{equation}
\end{Lemma}
\begin{proof} The proof is presented in Appendix~\ref{ap:lem1}.\end{proof}
To solve (\ref{eq4}) for $a_i$, we can further assume that $a_i = Ci^{-
\beta}$ \cite{complex}. Therefore, $\frac{a_{i+\alpha}}{a_i} =
\left(\frac{i+\alpha}{i}\right)^{-\beta}$
\begin{align}
\left(1-\frac{\beta}{i}\right)\Big(1 +
(1+\pdel)\big(1+\p(i+1)\big)+\pdel(i+2)\Big) =& (1+\pdel)(1 + pi +\p
^{\beta-1})\nonumber\\ &+ {\rm O}\left(\frac{1}{i}\right)
+\pdel(i+2)\left(1-\frac{2\beta}{i}\right).
\end{align}

Neglecting the ${\rm O}\left(\frac{1}{i}\right)$ terms, yields

\begin{align}
\label{rec}
(1+\pdel)( \p^{\beta - 1 }+\p\beta - \p ) =1 + \beta \pdel.
\end{align}

Note that the proof presented above depends on few assumptions. To give a
rigorous proof, the succeeding steps should be followed as described in
\cite{complex}:
\begin{itemize}
\item First, we need to show that the limit
$\lim_{\tim\rightarrow\infty}\frac{1}{\tim}\mathbf{E}\left\{f_\tim(i)
\right\} $ exists.
\item Let $a_i$ be the solution of (\ref{eq4}) such that
$\sum_{i=1}^{\infty}a_i = 1$ and $a_0 = 0$, then it is needed to show
that
\begin{equation}\label{eq58}
\lim_{\tim\rightarrow\infty}\frac{1}{\tim}\mathbf{E}\left\{f_\tim(i)
\right\}  =a_i.
\end{equation}
\item Finally, we should show that $a_i$ is proportional to $i^{-\beta}$,
where $\beta$ is the root of (\ref{rec}).
\end{itemize}

To complete the proof we define new function as follows $h_\tim(i) =
\frac{1}{\tim}\sum_{k=1}^{i}\mathbf{E}\{f_\tim(k)\}$ which can be
described as CDF of degree of each node in random graph. It is sufficient
to show that for all $i > 0$, \begin{equation} \label{eq59}
\lim_{\tim\rightarrow\infty}h_\tim(i) = \sum_{k=1}^ia_k
\end{equation}
where $a_i$ is the solution of (\ref{eq4}). It is obvious if (\ref{eq59})
holds, $h_\tim(i) - h_\tim(i-1) = a_i$ and thus
$$\lim_{\tim\rightarrow\infty}\frac{1}{\tim}\mathbf{E}\left\{f_\tim(i)
\right\}  =a_i$$ (as presented in (\ref{eq58})). The following lemma gives a
recurrence formula  to compute the value of $h(\tim+1,i)$.

\begin{Lemma}\label{lem6}
\begin{equation}
\label{lemm6}
h_{\tim+1}(i) = D_{\tim+1}(i)h_\tim(i) + B_{\tim+1}(i)h_\tim(i-1) +
C_{\tim+1}(i)h_\tim(i+1) + \frac{1+\pdel}{{\tim+1}}\sum_{j \geq i-
1}h_\tim(j)F(j,i-1,\p),
\end{equation}
where
\begin{eqnarray}
D_{\tim+1}(i) &=& \left(\frac{\tim -
\Big(\pdel(i+2)+(1+\pdel)\big(pi+1\big)\Big)}{{\tim+1}}\right),\nonumber\\
B_{\tim+1}(i) &=& \frac{(1+\pdel)(1 + pi)}{{\tim+1}},\nonumber\\
C_{\tim+1}(i) &=& \frac{\pdel(i+1)}{{\tim+1}},\nonumber\\
F(j,i,\p) &=& \sum_{k=0}^{i}{j\choose k}\p^k(1-\p)^{j-k} -
\sum_{k=0}^{i}{j+1\choose k}\p^k(1-\p)^{j+1-k}\nonumber.
\end{eqnarray}
\\
\end{Lemma}
This lemma can be proved by induction. The complete proof can be found in
Appendix~\ref{A7}. The recursive equation presented in Lemma~\ref{lem6}
is used later to prove that the degree distribution converges to a power
law.

\begin{Lemma}
Let $s_i = \sum_{k=1}^ia_i$ and  \begin{equation}\label{eq60} \omega(\tim)
= \sup_{i \geq 1}\frac{h_\tim(i)}{s_i}, \end{equation} where
$h_\tim(i)$ satisfies (\ref{lemm6}). Then the limit
$\lim_{\tim\rightarrow \infty}\omega(\tim)$ exists and we have
$\lim_{\tim\rightarrow \infty}\omega(\tim) = 1$.
\end{Lemma}

\paragraph*{Sketch of the proof} Knowing that $h_\tim(i)$ satisfies the
recurrence formula (\ref{lemm6}), the proof is similar to \cite{complex}.
Plugging $i=\tim$ in (\ref{eq60}) yields $\omega(\tim) \geq
\frac{h_\tim(\tim)}{s_\tim} \geq \frac{1}{s_\tim} \geq 1$. Using the
Lemma~\ref{lem6} and similar to \cite{complex}, it can be shown that
$\omega(\tim+1) \leq \omega(\tim)$. $\omega(\tim)$ is bounded and
decreasing, so the limit of $\lim_{\tim\rightarrow\infty}
\omega(\tim)$ exists. To show  $\lim_{\tim\rightarrow \infty}\omega(\tim)
= 1$, we assume that $\lim_{\tim\rightarrow\infty} \omega(\tim) = c$. It
can be shown that if $c \neq 1 $, $\omega(\tim) \leq 1 $ is violated.
Thus $c = 1$ and the proof is complete.

\end{proof}
\subsection{Proof of Lemma \ref{lem1}}\label{ap:lem1}
\begin{proof}
\begin{align}
\sum_{j\geq i}\frac{a_j}{a_i}{j\choose i}\p^i(1-\p)^{j-i} & = \sum_{j\geq
i}(\frac{i}{j})^{\beta}{j\choose i}\p^i(1-\p)^{j-i}\nonumber\\
&=  \sum_{j\geq i}(\frac{i}{j})^{\beta}{j\choose j-i}\p^i(1-\p)^{j-
i}\nonumber\\
&=  \left(1+O(\frac{1}{i})\right)\sum_{j\geq i}{j - \beta\choose j-
i}\p^i(1-\p)^{j-i}\nonumber\\
&=  \left(1+O(\frac{1}{i})\right)\p^i\sum_{k=0}{k+i - \beta\choose k}(1-
\p)^{k}\nonumber\\
&=  \left(1+O(\frac{1}{i})\right)\p^i\sum_{k=0}{\beta-i-1\choose k}(-
1)^k(1-\p)^{k}\nonumber\\
&=  \left(1+O(\frac{1}{i})\right)\p^ip^{\beta-i-1} =
\left(1+O(\frac{1}{i})\right)\p^{\beta-1}.
\end{align}
\end{proof}

\subsection{Proof of Lemma \ref{lem6}}\label{A7}
We prove the lemma by induction on $i$:
\paragraph*{For $i=1$} It is sufficient to show that: \\$h(\tim+1,1) =
D_{\tim+1}(1)h(\tim,1) + C_{\tim+1}(1)h(\tim,2) +
\frac{1}{\tim+1}\sum_{j\geq1}h(\tim,j)F(j,0,\p)$. Also using the
definition of $F(j,i,\p)$, we can rewrite $F(j,0,\p)$ as $(1-\p)^j-(1-
\p)^{j+1}$.
The number of nodes with degree one at time $\tim+1$ can be written as
following
\begin{align}
\label{eq61}
\mathbf{E}\{f(\tim+1,1)\} =& \left(1 -
\frac{(1+\pdel)(1+\p)+\pdel}{\tim}\right)\mathbf{E}\{f_\tim(1)\} +
\frac{2\pdel}{\tim}\mathbf{E}\{f_\tim(2)\} \nonumber\\&+ (1+\pdel)\sum_{j
\geq 1}\frac{1}{\tim}\mathbf{E}\{f_\tim(j)\}(1-\p)^j.
\end{align}
Note that (\ref{eq61}) is slightly different from the general equation
for each $i$, (\ref{eq3}). Because as described in Sec.\ref{sec:intro},
neighbors of a node with degree one cannot be eliminated from the graph
to maintain the connectivity in the graph. Therefore, a node with degree
one can change in the deletion step if that node is selected in the
deletion step (with probability $\pdel$). Using (\ref{eq61}),
$h(\tim+1,1)$ can be written as
\begin{align}
\label{eq62}
h(\tim+1,1) &= \frac{1}{{\tim+1}}\mathbf{E}\{f(\tim+1,1)\}\nonumber\\
&= \frac{1}{{\tim+1}}\left( \left(1 -
\frac{(1+\pdel)(1+\p)+\pdel}{\tim}\right)\mathbf{E}\{f_\tim(1)\} +
\frac{2\pdel}{\tim}\mathbf{E}\{f_\tim(2)\}\right)\nonumber\\
&\ +\frac{1}
{\tim+1}\sum_{j \geq 1}\frac{1+\pdel}{\tim}\mathbf{E}\{f_\tim(j)\}(1-\p)^j.
\end{align}

We know that $h(\tim,0) = 0$ for all $\tim$. Using the definition of
$h(\cdot,\cdot)$ and (\ref{eq61}), (\ref{eq62}) can be re-arranged as
follows
\begin{align}
\label{eq622}
h(\tim+1,1) =& \frac{1}{{\tim+1}}\bigg( \Big(\tim -
\big((1+\pdel)(1+\p)+\pdel\big)\Big)h(\tim,1) +
\frac{2\pdel}{\tim}\big(h(\tim,2)-h(\tim,1)\big)\nonumber\\ &+
(1+\pdel)\sum_{j \geq 1}(h(\tim,j) - h(\tim,j-1))(1-
\p)^j\bigg)\nonumber\\
 =&\frac{1}{{\tim+1}}\left( \Big(\tim - \big(3\pdel
+(1+\pdel)(1+\p)\big)\Big)h(\tim,1) + \frac{2\pdel}{\tim}h(\tim,2)
\right)\nonumber\\&+\frac{1+\pdel}{\tim +1}\sum_{j \geq 1}(h(\tim,j) -
h(\tim,j-1))(1-\p)^j
\end{align}
$\sum_{j \geq 1}(h(\tim,j) - h(\tim,j-1))(1-\p)^j$ can be written in
terms of the $F(j,i,\p)$.
\begin{align}
\label{eq63}
\sum_{j \geq 1}(h(\tim,j) - h(\tim,j-1))(1-\p)^j &=\sum_{j \geq
1}h(\tim,j)(1-\p)^j - \sum_{j \geq 1}(h(\tim,j-1)(1-\p)^j\nonumber\\
 &=\sum_{j \geq 1}h(\tim,j)(1-\p)^j - \sum_{j \geq 1}(h(\tim,j)(1-
\p)^j+1\nonumber\\
 &=\sum_{j \geq 1}h(\tim,j)\left((1-\p)^j - (1-\p)^{j+1}\right)
\nonumber\\
 &= \sum_{j\geq1}h(\tim,j)F(j,0,\p).
\end{align}
Substituting (\ref{eq63}) in (\ref{eq622}) yields
\begin{align}
\label{eq64}
h(\tim+1,1) &= \frac{1}{{\tim+1}}\left( \Big(\tim -
\big((1+\pdel)(1+\p)+3\pdel\big)\Big)h(\tim,1) +
\frac{2\pdel}{\tim}h(\tim,2)
+(1+\pdel)\sum_{j\geq1}h(\tim,j)F(j,0,\p)\right)\nonumber\\
&= D_{\tim+1}(1)h(\tim,1) + C_{\tim+1}(1)h(\tim,2) +
\frac{1+\pdel}{{\tim+1}}\sum_{j\geq1}h(\tim,j)F(j,0,\p).
\end{align}
Thus (\ref{lemm6}) holds for $i=1$. Now it is assumed that (\ref{lemm6})
holds for $i = k$, we want to show that it also holds for $i = k+1$.
\begin{align}
\label{eq65}
\mathbf{E}\{f(\tim+1,k+1)\} =&\left(1 -
\frac{\pdel(k+2)+(1+\pdel)\big(\p(k+1)+1\big)}{\tim}\right)\mathbf{E}\{f(
\tim,k+1)\}\nonumber\\&+\left(\frac{(1+\pdel)(1+pk)}{\tim}\right)\mathbf{
E}\{f_\tim(k)\}
+\left(\frac{\pdel(k+2)}{\tim}\right)\mathbf{E}\{f_\tim(k+2)\}\nonumber\\
&+ (1+\pdel)\sum_{j \leq k}\frac{f_\tim(j)}{\tim}{j \choose k} \p^k(1 -
\p)^{j-k}.
\end{align}
from definition of $h(\tim,k)$, we have : $\mathbf{E}\{f_\tim(k)\} =
\tim\left(h(\tim,k) - h(\tim,k-1)\right)$. Eq. (\ref{eq65}) can be re-
written as follows
\begin{align}
\label{eq66}
\mathbf{E}\{f(\tim+1,k+1)\} =&\left(\tim -
\Big(\pdel(k+2)+(1+\pdel)\big(\p(k+1)+1\big)\Big)\right)\big(h(\tim,k+1)
- h(\tim,k)\big)\nonumber\\
&+(1+\pdel)(1+pk)\big(h(\tim,k) - h(\tim,k-
1)\big)+\pdel(k+2)\big(h(\tim,k+2) -
h(\tim,k+1)\big)\nonumber\\&+(1+\pdel) \sum_{j \leq k}\big(h(\tim,j) -
h(\tim,j-1)\big){j \choose k} \p^k(1 - \p)^{j-k}.
\end{align}
Using the Abel summation identity, and knowing that
$$F(j,k,\p)=\sum_{k=0}^{k}{j\choose k}\p^k(1-\p)^{j-k} -
\sum_{k=0}^{k}{j+1\choose k}\p^k(1-\p)^{j+1-k},$$ the last term can be
written as
\begin{align}\label{eq67}
& \sum_{j \leq k}\big(h(\tim,j) - h(\tim,j-1)\big){j \choose k} \p^k(1 -
\p)^{j-k} \\
& \quad =\sum_{j \geq k}\left({j \choose k} \p^k(1 - \p)^{j-k} - {j+1
\choose k} \p^k(1 - \p)^{j+1-k}\right)
%\nonumber\\ &
-\p^kh(\tim,k-
1)\nonumber\\
 &\quad = -\p^kh(\tim,k-1)+\sum_{j \geq k}h(\tim,j)\big(F(j,k,\p)-
F(j,k-1,\p)\big).
\end{align}
Substituting (\ref{eq67}) in (\ref{eq66}) yields
\begin{align}
\label{eq68}
\mathbf{E}\{f(\tim+1,k+1)\} =& h(\tim,k+2)(\pdel(k+2)) +
h(\tim,k+1)\bigg(\tim -
\big(2\pdel(k+2)+(1+\pdel)(\p(k+1)+1)\big)\bigg)\nonumber\\
&+h(\tim,k)\Big((1+\pdel)\big(2 + \p(2k+1)\big)\nonumber\\
& + \pdel(k+2) - \tim\Big)
+ h(\tim,k-1)(1+\pdel)(-1-pk-\p^k)\nonumber\\
&+(1+\pdel)\sum_{j \geq k}h(\tim,j)\big(F(j,k,\p)- F(j,k-1,\p)\big).
\end{align}
The value of $h(\tim+1,k+1)$ can be computed using $h(\tim,k+1)$ and
$\mathbf{E}\{f_\tim(k+1)\}$ as follows
\begin{equation}
\label{eq69}
h(\tim+1,k+1) = h(\tim+1,k) + \frac{1}{\tim+1}\mathbf{E}\{f(\tim+1,k+1)\}.
\end{equation}
Eq.(\ref{eq68}) gives an expression for $\mathbf{E}\{f(\tim+1,k+1)\}$ in
terms of the value of $h(\cdot,\cdot)$ at time $\tim$. Substituting
(\ref{eq68}) in (\ref{eq69}) gives a recursive equation for computing
$h(\tim+1,k+1)$:
\begin{align}
\label{eq70}
h(\tim+1,k+1) =& h(\tim+1,k) +
\frac{1}{{\tim+1}}\mathbf{E}\{f(\tim+1,k+1)\}\nonumber\\
=& D_{\tim+1}(k)h(\tim,k) + B_{\tim+1}(k)h(\tim,k-1) +
C_{\tim+1}h(\tim,k+1)\nonumber\\
&+ \frac{1+\pdel}{{\tim+1}}\sum_{j \geq k-
1}h(\tim,j)F(j,k-1,\p)\nonumber\\
&+\frac{1}{\tim+1}\Bigg( h(\tim,k+2)(\pdel(k+2)) + h(\tim,k+1) \nonumber\\
& \bigg(\tim
-\big(2\pdel(k+2)+(1+\pdel)(\p(k+1)+1)\big)\bigg)\nonumber\\
&+h(\tim,k)\Big((1+\pdel)\big(2 + \p(2k+1)\big) + h(\tim,k-1)(1+\pdel)(-
1-pk-\p^k)\nonumber\\
&+(1+\pdel)\sum_{j \geq k}h(\tim,j)\big(F(j,k,\p)- F(j,k-1,\p)\big)\Bigg).
\end{align}
We assume that (\ref{eq3}) holds for $i = k$ so substituting the values
for $D_{\tim+1}(k)$, $B_{\tim+1}(k)$, and $C_{\tim+1}(k)$ from (\ref{eq3})
in (\ref{eq70}) yields
\begin{align}
\label{eq71}
h(\tim+1,k+1) =& h(\tim,k+2)\left(\frac{\pdel(k+2)}{\tim+1}\right) +
h(\tim,k+1)\left(\frac{\tim -
\Big(\pdel(k+3)+(1+\pdel)\big(\p(k+1)+1\big)\Big)}{{\tim+1}}\right)
\nonumber\\
&+h(\tim,k)\left(\frac{(1+\pdel)\big(1 +
\p(k+1)\big)}{\tim+1}\right)+\frac{1+\pdel}{\tim+1} \sum_{j \geq
k}h(\tim,j)\big(F(j,k,\p)\big).
\end{align}
(\ref{eq71})can be written as follows
\begin{align}
h(\tim+1,k+1) =& D_{\tim+1}(k+1)h(\tim,k+1) + B_{\tim+1}(k+1)h(\tim,k) +
C_{\tim+1}(k+1)h(\tim,k+2) \nonumber\\&+ \frac{1+\pdel}{{\tim+1}}\sum_{j
\geq k}h(\tim,j)F(j,k,\p).
\end{align}
Thus, (\ref{eq3}) holds for $i = k+1$ and the proof is completed by
induction.

 \subsection{Proof of Theorem \ref{theo2}}\label{ap:bound}
 \begin{proof}
 Define the Liapunov function $V(x) = (x'
x)/2$ for  $x \in \mathbb{R}^\s$.
%The following expression can be
%written for the growth of the difference between the sample path and the
%expected distribution.
Use $\Et$ to denote the conditional expectation
with respect to the $\sigma$-algebra, $\mathcal{H}_n$, generated by
$\{\obs_j,\mc_j, \quad j\leq \tim\}$.

\begin{align}
\label{eq17}
\Et\{V(\tg_{\tim+1})-V(\tg_\tim)\} =& \Et\Big\{\tg'_\tim[-
\esa\tg_\tim+\esa\left({\obs_{\tim+1}}-\mathbf{E}\{\bg(\mc_\tim)\}\right)
+\mathbf{E}\{\bg(\mc_\tim) - \bg(\mc_{\tim+1})\}]\Big\}\nonumber\\
&+\Et\Big\{|-\esa\tg_\tim+\esa\left(\obs_{\tim+1}-
\mathbf{E}\{\bg(\mc_\tim)\}\right)+\mathbf{E}\{\bg(\mc_\tim) -
\bg(\mc_{\tim+1})\}|^2\Big\},
\end{align}
where $\obs_{\tim+1}$ and $\bg(\mc_\tim)$ are  vectors in
$\mathbb{R}^\s$ with elements $\obs_\tim(i)$ and $\bg(\mc_\tim,i)$,\quad
$1 \leq i \leq \s$, respectively.
It is easily seen that
\beq
\Et\{\bg(\mc_\tim) - \bg(\mc_{\tim+1})\} = O(\emc),
\ %\eeq
%and
%\beq
\Et\{|\bg(\mc_\tim) - \bg(\mc_{\tim+1})|^2\} = O(\emc).
\eeq

Using $K$ to denote a generic positive value
(with the notation $KK=K$  and $K+K =K$),
a farmiliar inequality $ab \leq \frac{a^2+b^2}{2}$ yields
\beq\label{eq:elem}  O(\esa\emc) = O(\esa^2 + \emc^2).\eeq  Moreover
 we have $|\tg_\tim| = |\tg_\tim| \cdot 1 \leq (|\tg_\tim|^2
+1 )/2$. Thus \beq\label{eq:pro} O(\emc)|\tg_\tim|\leq
O(\emc)\left(V(\tg_\tim) + 1\right). \eeq
Then detailed estimates lead to
\begin{align}
%\label{eq20}
\label{eq21}
\Et\Big\{\Big|-\esa\tg_\tim&+\esa\left(\obs_{\tim+1}-
\mathbf{E}\{\bg(\mc_\tim)\}\right)+\mathbf{E}\{\bg(\mc_\tim) -
\bg(\mc_{\tim+1})\}\Big|^2\Big\}
% \nonumber\\
%&\le K \Et\Bigg\{\esa^2|\tg_\tim|^2+\esa^2|(\obs_{\tim+1}-
%\mathbf{E}\{\bg(\mc_\tim)\}|^2 +\esa^2
%\big|\tg'_\tim\mathbf{E}\big\{\obs_{\tim+1}-
%\mathbf{E}\{\bg(\mc_{\tim+1})\}\big\}\big|\nonumber\\
%&+ \esa |\tg'_\tim\mathbf{E}\{\bg(\mc_\tim) - \bg(\mc_{\tim+1})\}| +
%\esa|\left(\obs_{\tim+1}-
%\mathbf{E}\{\bg(\mc_\tim)\}\right)'\mathbf{E}\{\bg(\mc_\tim)\nonumber\\
%&- \bg(\mc_{\tim+1})\}|\Bigg\} + \mathbf{E}\{|\bg(\mc_\tim) -
%\bg(\mc_{\tim+1})|\}^2
%\end{align}
%It follows that
%\begin{align}
%\label{eq21}
%\Et\Big\{\Big|-\esa\tg_\tim+\esa\left(\obs_{\tim+1}-
%\mathbf{E}\{\bg(\mc_\tim)\}\right)+\mathbf{E}\{\bg(\mc_\tim) -
%\bg(\mc_{\tim+1})\}\Big|^2\Big\}
= O(\esa^2+\emc^2)(V(\tg_\tim)+1)
\end{align}
%An upper bound for the first term in the RHS of (\ref{eq17}) can be
%derived as
%\begin{align}
%\label{eq22}
%\Et\Big\{\tg'_\tim[-\esa\tg_\tim+\esa\left(\obs_{\tim+1}-
%\mathbf{E}\{\bg(\mc_\tim)\}\right) &+\mathbf{E}\{\bg(\mc_\tim) -
%\bg(\mc_{\tim+1})\}]\Big\}
%= -2\esa V(\tg_\tim) \nonumber\\&+ \esa \Et\{\tg'_\tim[\obs_{\tim+1} -
%\mathbf{E}\bg(\mc_\tim)]\} + \Et\{\tg'_\tim\mathbf{E}[\bg(\mc_{\tim+1}) -
%\bg(\mc_\tim)]\}.
%\end{align}
%Substituting (\ref{eq21}) and (\ref{eq22}) in (\ref{eq17}) yields
Furthermore, wee obtain that
\begin{align}
\label{eq23}
\Et\{V(\tg_{\tim+1})-V(\tg_\tim)\} =&
-2\esa V(\tg_\tim) + \esa \Et\{\tg'_\tim[\obs_{\tim+1} -
\mathbf{E}\bg(\mc_\tim)]\} \nonumber\\&+
\Et\{\tg'_\tim\mathbf{E}[\bg(\mc_{\tim+1}) - \bg(\mc_\tim)]\}
+ O(\esa^2 +\emc^2)(V(\tg_\tim)+1).
\end{align}

Define $V^{\emc}_1$ and $V^{\emc}_2$ as following
\begin{align}
V^{\emc}_1(\tg,\tim) &= \esa\sum_{j=\tim}^{\infty}\tg'\Et\{\obs_{j+1} -
\mathbf{E}\bg(\mc_j)\},\nonumber\\
V^{\emc}_2(\tg,\tim) &= \sum_{j=\tim}^{\infty}\tg'\mathbf{E}_n\{\bg(\mc_j)-
\bg(\mc_{j+1})\},
\end{align}
It can be shown that
\bq{eq:order-vp} \barray
\ad |V^\rho_1(\tilde g, n)| = O(\e) (V(\tilde g)+ 1),\\
\ad |V^\rho_2(\tilde g, n)| = O(\rho) (V(\tilde g)+ 1).\earray\eq
Define $ W(\tg,\tim)$ as
\begin{equation}
 W(\tg,\tim) = V(\tg) +  V^{\emc}_1(\tg,\tim)+V^{\emc}_2(\tg,\tim).
 \end{equation}
This leads to
%To average out the second and the third terms in the RHS of the equation
%(\ref{eq23}), the expectation value of the growth of $W(\tg,\tim)$ is
%evaluated.
\begin{align}
\label{eq27}
\Et\{W(\tg_{\tim+1},\tim+1) - W(\tg_\tim,\tim)\} =&\Et\{V(\tg_{\tim+1})-
V(\tg_\tim)\} + \Et\{V^{\emc}_1(\tg_{\tim+1},\tim+1)-
V^{\emc}_1(\tg_\tim,\tim)\}\nonumber\\&+\Et\{V^{\emc}_2(\tg_{\tim+1},\tim
+1)- V^{\emc}_2(\tg_\tim,\tim)\}.
\end{align}
%Substituting (\ref{eq23}), (\ref{eq25}), and (\ref{eq26}) in the equation
%(\ref{eq27}) yields
Moreover,
\begin{align}
\label{eq33}
\Et\{W(\tg_{\tim+1},\tim+1) - W(\tg_\tim,\tim)\} =-2\esa V(\tg_\tim)+
O(\esa^2 +\emc^2)(V(\tg_\tim)+1).
\end{align}
Eq. (\ref{eq33}) can be rewritten as
\begin{align}
\label{eq30}
\Et\{&W(\tg_{\tim+1},\tim+1) - W(\tg_\tim,\tim)\} \leq -2\esa
W(\tg_\tim,\tim)+ O(\esa^2 +\emc^2)(W(\tg_\tim,\tim)+1).
\end{align}
If $\esa$ and $\emc$ are chosen small enough, then there exists an small
$\lambda$ such that $-2\esa + O(\emc^2)+ O(\esa^2) \leq -\lambda\esa$. So
(\ref{eq30}) can be re-arranged to the following,
\begin{align}
\Et\{W(\tg_{\tim+1},\tim+1)\leq (1-\lambda\esa) W(\tg_\tim,\tim)+
O(\esa^2 +\emc^2).
\end{align}
Taking expectation of both sides yields
\begin{align}
\label{eq32}
\mathbf{E}\{W(\tg_{\tim+1},\tim+1)\}\leq
(1-\lambda\esa) \mathbf{E}\{W(\tg_\tim,\tim)\}+O(\esa^2 +\emc^2).
\end{align}
Iterating on (\ref{eq32}) yields
\begin{align}
\mathbf{E}\{W(\tg_{\tim+1},\tim+1)\}\leq
(1-\lambda\esa)^{\tim-N_\emc}\mathbf{E}\{W(\tg_{N_\emc},N_\emc)\}+\sum_{j
= N_\emc}^{\tim}O(\esa^2 +\emc^2)(1-\lambda\esa)^{j-N_\emc},
\end{align}
so
%$\sum_{j = N_\emc}^{\tim}O(\esa^2 +\emc^2)(1-\lambda\esa)^{j-N_\emc}$ can
%be simplified to $O(\esa^2 +\emc^2)\sum_{j = N_\emc}^{\tim}(1-
%\lambda\esa)^{j-N_\emc} = O(\esa^2 +\emc^2)\frac{1+O(\esa)}{\lambda\esa}
%= O(\esa+ \emc^2/\esa)$.
\begin{align}
\mathbf{E}\{W(\tg_{\tim+1},\tim+1)\}\leq
(1-\lambda\esa)^{\tim-
N_\emc}\mathbf{E}\{W(\tg_{N_\emc},N_\emc)\}+O\left(\esa
+\frac{\emc^2}{\esa}\right).
\end{align}
If $\tim$ is large enough we can approximate $(1-\lambda\esa)^{\tim-
N_\emc} = O(\esa)$
\begin{equation}
\mathbf{E}\{W(\tg_{\tim+1},\tim+1)\}\leq
O\left(\esa+\frac{\emc^2}{\esa}\right)
\end{equation}
Finally, using (\ref{eq:order-vp}) and replacing
$W(\tg_{\tim+1},\tim+1)$ with $V(\tg_{\tim+1})$, we obtain
\begin{equation}
\mathbf{E}\{V(\tg_{\tim+1})\}\leq
O\left(\emc+\esa+\frac{\emc^2}{\esa}\right).
\end{equation}
\end{proof}

\subsection{Sketch of the Proof of Theorem \ref{theo3}}\label{ap:conv-pf}
\noindent
Since the proof is similar to \cite[Theorem 4.5]{YKI04}, we only
indicate the main steps needed and omit most of the vabatim details.

(1) First we show that the two component process $(\hat
g^\e\cd,\theta^\e\cd)$ is tight in $D([0,T]:
%\rr^S
\rr^{N_0}\times \M)$.
%[?? changed to $\rr^{N_0}$, ok??]
Using the techniques as in \cite[Theorem 4.3]{YinZ05}, it can be shown that
$\theta^\e\cd$ converges weakly to a continuous-time Markov chain
generated by $Q$. Thus, we mainly need to consider $\hat g^\e\cd$.
We show that $$\lim_{\Delta\to 0}\limsup_{\e\to 0} \bE [\sup_{0\le s \le
\Delta} \bE^\e_t | \hat
g^\e(t+s)-  \hat g^\e(t)|^2] =0,$$
where $\bE^\e_t$ denotes the conditioning on the past information up to $t$.
Then the tightness follows from the criterion \cite[p. 47]{Kushner84}.

(2) Since $(\hat g^\e\cd,\theta^\e\cd)$ is tight, we can extract
weakly convergent subsequence according to the Prohorov theorem (see
\cite{KY03}).
To figure out the limit, we show that $(\hat g^\e\cd, \theta^\e\cd)$ is a
solution of the martingale problem with operator
$L_0$. For each $i\in \M$ and
continuously differential function
with compact support $f(\cdot, i)$, the operator is given by
\bq{l1-def}
L_0 f(\hat g,i)= \nabla f'(\hat g,i) [- \hat g
 +\bar g(i)] + \sum_{j\in \M} q_{ij}
f(\hat g,j), \ i\in \M.\eq
We can further demonstrate the martingale problem with operator $L_0$
has a unique solution in the sense in distribution. Thus the desired
convergence property follows.

\subsection{Sketch of the Proof of Theorem \ref{theo4}}\label{ap:sde-pf}
(1) First note
\bq{nu-defn}\nu_{n+1}=\nu_n - \e \nu_n +\sqrt \e (y_{n+1}-\bE \bar
g(\theta_n) )
+ { \bE [\bar g(\theta_n)- \bar g(\theta_{n+1}]\over \sqrt \e }.\eq
Again, the approach is similar to that of
\cite[Theorem 5.6]{YKI04}. So again, we will be brief.

(2) Define an operator
 \bq{op-def-sde-0}
 {\cal L} f(\nu,i)=-\nabla f'(\nu,i) \nu  + {1\over 2}
 \tr [\nabla ^2
 f(\nu,i)
 \Sigma(i) ] +\sum_{j\in \M} q_{ij}
f(\nu,j), \ i\in \M,\eq
 for function $f(\cdot,i)$ that has continuous partial derivatives
 with respect to $\nu$ up to the second order and that has compact
 support. It can be show that the associated martingale problem has
 a unique solution in the sense in distribution.

(3) It is natural now to work with a truncated process.
For a fixed but otherwise arbitrary  $r_1>0$,
 define a truncation function
$$q^{r_1}(x)=\left\{ \barray 1, & \hbox{ if } x \in S^{r_1},\\
0, & \hbox{ if } x \in \rr^{N_0}- S^{r_1},\earray \right.$$
where $S^{r_1}= \{ x \in \rr^{N_0}: |x| \le r_1\}$.
Then we get the truncated iterates
\bq{trun-it}
\nu^{r_1}_{n+1}=\nu^{r_1}_n - \e \nu^{r_1}_n q^{r_1}(\nu^{r_1}_n)
 +\sqrt \e (y_{n+1}-\bE \bar g(\theta_n) )
 + { \bE [\bar g(\theta_n)- \bar g(\theta_{n+1}]\over \sqrt \e }q^{r_1}(\nu^{r_1}_n).\eq
 Define $\nu^{\e,r_1}(t)= \nu^{r_1}_n$ for $t\in [\e n, \e n+\e)$. Then
 $\nu^{\e,r_1}\cd$ is an $r$-truncation of
 $\nu^\e\cd$; see \cite[p. 284]{KY03} for a definition.
 We then show the truncated process $(\nu^{\e,r_1}\cd, \theta^\e\cd)$
 is
 tight. Moreover, by Prohorov's theorem, we can extract a convergent
 subsequence with limit $(\nu^{r_1}\cd, \theta\cd)$
 such that
 the limit $(\nu^{r_1}\cd,\theta\cd)$ is the solution of the martingale
 problem with operator
 ${\cal L}^{r_1}$ defined by
 \bq{op-def-sde}
 {\cal L}^{r_1} f^{r_1}(\nu,i)=-\nabla f^{r_1,\prime}(\nu,i) \nu  + {1\over 2}
 \tr [\nabla ^2
 f^{r_1}(\nu,i)
 \Sigma(i) ]+\sum_{j\in \M} q_{ij}
f^{r_1}(\nu,j), \ i\in \M,\eq
 where $f^{r_1}(\nu,i)= f(\nu,i) q^{r_1}(\nu)$.

(4) Letting $r_1\to \infty$, we show that the un-truncated process also
converges and the limit denoted by $(\nu\cd,\theta\cd)$ is precisely
the martingale problem with operator ${\cal L}$ defined in
(\ref{op-def-sde-0}). Furthermore, the limit covariance can be
evaluated as in \cite[Lemma 5.2]{YKI04}.

\bibliographystyle{IEEEtran}
\bibliography{reff}

\end{document}